\documentclass[a4paper,onecolumn,10pt,accepted=2026-09-04]{quantumarticle}
\pdfoutput=1
\usepackage[utf8]{inputenc}
\usepackage[english]{babel}
\usepackage[T1]{fontenc}
\usepackage{amsmath}
\usepackage{mathrsfs}
\usepackage{bm}
\usepackage{dsfont}
\usepackage{physics}
\usepackage{array}
\usepackage{subcaption} 
\usepackage{tikz}
\usepackage{lipsum}
\usepackage{times,amsmath,amsfonts,amssymb,latexsym,amsthm,bbm}
\usepackage{graphicx,epsf}
\usepackage[toc,page,header]{appendix}
\usepackage{minitoc}
\usepackage[numbers,sort&compress]{natbib}
\usepackage{color}
\theoremstyle{plain}
\newtheorem{thm}{\protect\theoremname}

\theoremstyle{plain}
\newtheorem{cor}{\protect\corollaryname}

\usepackage[english,italian]{varioref}                          
\usepackage{indentfirst}      
\usepackage{ulem}
\usepackage{amsmath,amssymb}           
\usepackage{caption}
\usepackage{mathtools}
\usepackage{subcaption}
\usepackage{graphicx,epsf}
\usepackage{physics}
\usepackage{comment}
\usepackage{verbatim}
\usepackage{textcomp}       

\usepackage{booktabs, palatino}
\usepackage{amsthm}                              
\usepackage{stmaryrd, bbm}                         
\usepackage{cancel}
\usepackage[toc]{appendix}

\usepackage[colorlinks]{hyperref}

\hypersetup{colorlinks=true, linkcolor=black, urlcolor=blue, citecolor=blue}
\theoremstyle{definition}                           

\newtheorem{fact}{Fact}
\theoremstyle{plain}
\newtheorem*{teorema*}{Theorem}
\newtheorem{proposizione}{Proposition}            

\theoremstyle{remark}                                       

\usepackage[colorlinks]{hyperref}
\hypersetup{colorlinks=true, linkcolor=black, urlcolor=blue, citecolor=blue}

\newcommand{\be}{\begin{equation}}
\newcommand{\ee}{\end{equation}}
\newcommand{\bs}{\begin{split}}
\newcommand{\es}{\end{split}}

\newcommand{\ot}{\otimes}

\allowdisplaybreaks
\newcommand{\ba}{\begin{eqnarray}}

\newcommand{\ea}{\end{eqnarray}}

\newcommand\stab{{\operatorname{STAB}}}
\newcommand\stabone{{\operatorname{STAB}_{1}}}
\newcommand\stabzero{{\operatorname{STAB}_{0}}}

\newcommand{\ignore}[1]{}

\newcommand{\pauli}[1]{\mathbb{P}_{#1}}

\newcommand{\Cl}{\mathcal{C}}

\DeclareMathOperator{\sign}{sign}

\newcommand{\Var}{\mathsf{Var}}

\providecommand{\corollaryname}{Corollary}
\providecommand{\theoremname}{Theorem}

\newcommand{\1}{\leavevmode{\rm 1\ifmmode\mkern  -4.8mu\else\kern -.3em\fi I}}
\renewcommand{\cal}{\mathcal}

\renewcommand{\tr}[1]{{\rm Tr}\left[#1\right]}

\begin{document}

\title{Non-stabilizerness and violations of CHSH inequalities}
\author{Stefano Cusumano}
\orcid{0000-0002-4379-3739}
\email[Corresponding author: ]{ste.cusumano@gmail.com}
\affiliation{Dipartimento di Fisica `Ettore Pancini', Universit\`a degli Studi di Napoli Federico II,
Via Cintia 80126,  Napoli, Italy}
\affiliation{INFN, Sezione di Napoli, Italy}

\author{Lorenzo Campos Venuti}
\orcid{0000-0002-0217-6101}
\affiliation{Dipartimento di Fisica `Ettore Pancini', Universit\`a degli Studi di Napoli Federico II,
Via Cintia 80126,  Napoli, Italy}
\affiliation{Department of Physics and Astronomy, University of Southern California, Los Angeles, USA}

\author{Simone Cepollaro}
\orcid{0009-0002-0309-3508}
\affiliation{Scuola Superiore Meridionale, Largo S. Marcellino 10, 80138 Napoli, Italy}
\affiliation{INFN, Sezione di Napoli, Italy}

\author{Immacolata De Simone}
\affiliation{Dipartimento di Fisica `Ettore Pancini', Universit\`a degli Studi di Napoli Federico II,
Via Cintia 80126,  Napoli, Italy}

\author{Gianluca Esposito}
\orcid{0000-0002-4483-1995}
\affiliation{Scuola Superiore Meridionale, Largo S. Marcellino 10, 80138 Napoli, Italy}
\affiliation{INFN, Sezione di Napoli, Italy}

\author{Daniele Iannotti}
\orcid{0009-0009-0738-5998}
\affiliation{Scuola Superiore Meridionale, Largo S. Marcellino 10, 80138 Napoli, Italy}
\affiliation{INFN, Sezione di Napoli, Italy}

\author{Barbara Jasser}
\orcid{0009-0004-5657-2806}
\affiliation{Scuola Superiore Meridionale, Largo S. Marcellino 10, 80138 Napoli, Italy}
\affiliation{INFN, Sezione di Napoli, Italy}

\author{Jovan Odavi\' c}
\orcid{0000-0003-2729-8284}
\affiliation{Dipartimento di Fisica `Ettore Pancini', Universit\`a degli Studi di Napoli Federico II,
Via Cintia 80126,  Napoli, Italy}
\affiliation{INFN, Sezione di Napoli, Italy}

\author{Michele Viscardi}
\orcid{0000-0002-1606-9622}
\affiliation{Dipartimento di Fisica `Ettore Pancini', Universit\`a degli Studi di Napoli Federico II,
Via Cintia 80126,  Napoli, Italy}
\affiliation{INFN, Sezione di Napoli, Italy}

\author{Alioscia Hamma}
\orcid{0000-0003-0662-719X}
\affiliation{Dipartimento di Fisica `Ettore Pancini', Universit\`a degli Studi di Napoli Federico II, Via Cintia 80126,  Napoli, Italy}
\affiliation{INFN, Sezione di Napoli, Italy}
\affiliation{Scuola Superiore Meridionale, Largo S. Marcellino 10, 80138 Napoli, Italy}

\begin{abstract}
We study quantitatively the interplay between entanglement and non-stabilizer resources in violating the CHSH inequalities. We show that, while non-stabilizer resources are necessary, they must have a specific structure, namely they need to be both asymmetric and (surprisingly) {\it local}. We employ stabilizer entropy (SE) to quantify the non-stabilizer resources involved and the probability of violation given the resources. We show how spectral quantities related to the flatness of entanglement spectrum and its relationship with non-local SE affect the CHSH inequality. Finally, we utilize these results - together with tools from representation theory - to construct a systematic way of building ensembles of states with higher probability of violation.
\end{abstract}

\maketitle

\section{Introduction}
It is a well-known fact that the Clauser-Horne-Shimony-Holt (CHSH)  inequality \cite{PhysRevLett.23.880} can be violated in quantum mechanics due to quantum entanglement. However, it has been recognized more recently that entanglement alone is not sufficient: resources beyond {\it stabilizerness} - colloquially known as {\it magic} - are also necessary~\cite{PhysRevA.91.042103, Howard2014,PhysRevA.85.022304, macedo2025witnessingmagicbellinequalities,1nkl-sphd}.

In this paper, we perform a theoretical and quantitative study of the interplay between the entangling and non-stabilizer resources necessary to violate the CHSH inequalities. To this end, we build a setting that does not  arbitrarily separate the resources involved in state preparation and measurement. 

Non-stabilizer resources are also necessary, together with entanglement, to attain quantum advantage~\cite{gottesman_stabilizer_1997,gottesman_heisenberg_1998,aaronson_improved_2004}. More generally, quantum complex behavior rises from the interplay of  both entangling and non-stabilizer resources~\cite{PhysRevB.96.020408, Leone2021quantumchaosis}. In order to quantify non-stabilizerness, we resort to the Stabilizer Entropy (SE), the unique computable monotone of non-stabilizerness for pure states~\cite{leone_stabilizer_2022, leone_stabilizer_2024}. SE is experimentally measurable~\cite{Oliviero2022} and efficiently computable by  tensor networks methods~\cite{Lami2023, tarabunga2024, TarTirChaDal23}. 

The availability of a computable non-stabilizerness measure such as SE , has allowed to test and quantify the role of non-stabilizer resources in several settings and scenarios, ranging from quantum phase transitions~\cite{Oliviero2021,10.21468/SciPostPhys.12.3.096,True2022transitionsin,catalano2024magicphasetransitionnonlocal,li2024measurementinducedmagicresources} and quantum chaos~\cite{Leone2021quantumchaosis,odavić2025stabilizerentropynonintegrablequantum,jasser2025stabilizerentropyentanglementcomplexity,cusumano2026probeschaoscliffordgroup,iannotti2026nonstabilizernessu1symmetrychaotic}, to high-energy physics~\cite{cao2024gravitationalbackreactionmagical,brokemeier2024quantum, White_White_2024}, quantum-information~\cite{PhysRevA.106.062434,PhysRevA.107.022429,PhysRevLett.132.080402,PhysRevA.109.022429,Leone2024learningtdoped, hou2025stabilizerentanglementmagichighway, gu2024magicinducedcomputationalseparationentanglement,Wang2023,PhysRevB.96.020408,10.21468/SciPostPhys.9.6.087,PhysRevA.109.L040401,iannotti2025entanglement,cepollaro2025stabilizerentropysubspaces}, and condensed matter~\cite{PhysRevA.106.042426, 10.21468/SciPostPhys.15.4.131, PhysRevA.108.042407, PRXQuantum.4.040317,PRXQuantum.3.020333, wei2025longrangenonstabilizernesstopologycorrelation,PhysRevD.109.126008,cepollaro2025harvestingstabilizerentropynonlocality,falcão2025magicdynamicsmanybodylocalized,PhysRevB.111.054301,ding2025evaluatingmanybodystabilizerrenyi,hoshino2025stabilizerrenyientropyconformal,iannotti2026nonlocalmagicresourcesfermionic}.

In this paper, we detail what kind of structure entanglement and SE need to have in order to violate the CHSH inequality. Counterintuitively, we prove that SE needs to be \textit{local} in order to obtain a violation: non-local SE \cite{cao2024gravitationalbackreactionmagical} is shown to be detrimental to CHSH violations, bounding the maximum expectation value of the Bell operator. Moreover, the resources must be {\it asymmetric} between Alice and Bob. We use both Haar averaging and numerical techniques to compute the probability of violations given the resources. Finally, we use the technique of isospectral twirling to show how knowledge of the structure of entangling and non-stabilizer resources can be used to improve the probability of a CHSH violation under imperfect control~\cite{PhysRevLett.104.050401,cieslinski_how_2025,cieslinski_analysing_2024} or in satellite experiments~\cite{bedington_progress_2017,qtc2.12015}.  

\section{Setting the stage}\label{setup}
The Horodeckis' seminal work \cite{Horodecki1995}
 establishes necessary and sufficient conditions on the state preparation in order to violate CHSH inequalities by means of {\it local} measurements. However, it was noted in \cite{PhysRevA.91.042103} that one would need non-Clifford measurements. Because of the state-effect duality, though, it is clear that one could just perform measurements that are neutral from the stabilizer resources point of view if one prepares the state with the required resources. Let us consider an example. We work in the Hilbert space of two qubits ${\mathcal H}_{AB} \equiv {\mathcal H}_{A} \otimes {\mathcal H}_{B}$ and define the `resource free' CHSH operator (we omit the tensor product symbol when not strictly necessary)
\ba
\label{eq:b0_def}
B_0 := X\otimes (X+Z)+Z\otimes(-X+Z)=XX+XZ-ZX+ZZ \, . 
\ea
 \begin{figure}[!t]
\centering
\includegraphics[width=\linewidth]{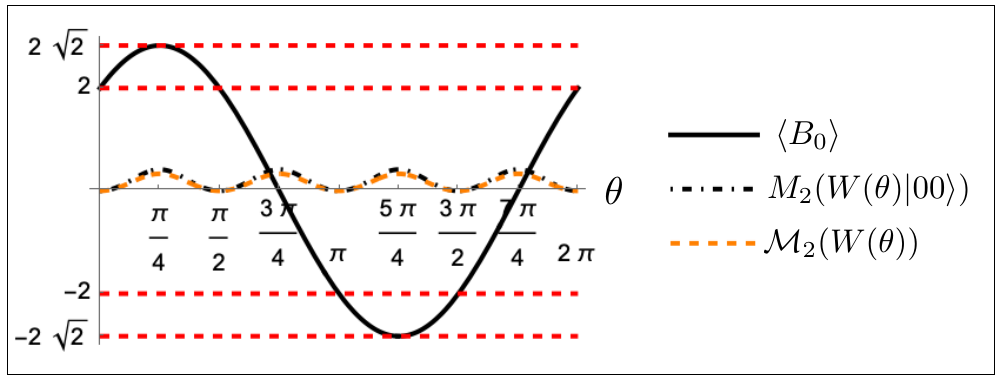}
\caption{Plot of the expectation value of the rotated operator $\langle B_0\rangle$ (solid black line),  the SE $M_2$ of the rotated Bell state $W(\theta)\ket{00}$ (dashed-dotted black line) and the non-stabilizing power $\mathcal{M}_2$ of the rotation operator (dashed orange line). The dashed red lines represent the values of violation of CHSH inequalities and the Tsirelson's bound.}
\label{fig:violations_theta}
\end{figure}
The above operator is resource free because it describes four measurements that are {\it both} local and within the stabilizer formalism, being Pauli measurements. A direct calculation shows that, even preparing a Bell state (say $\ket{\Phi^+}\equiv\ket{00}+\ket{11}/\sqrt{2}$), these measurements will not lead to a CHSH violation, as $\Tr(B_0\Phi^+)= 2$. On the other hand, if we prepare the state $R_y(\theta)  \otimes I\ket{\Phi^+}=\exp[-i\frac{\theta}{2}Y] \otimes I\ket{\Phi^+} = (R_y(\theta) \ot I)C_X (H\ot I) \ket{00}\equiv W (\theta) \ket{00}$, one can violate the CHSH inequalities for a certain range of the rotation angle $\theta$, see Fig.~\ref{fig:violations_theta}, even attaining the Tsirelson's bound for $\theta =\pi/4$. Equivalently, one could have prepared $\Phi^+$ and used the CHSH operator $R^\dag_y(\theta)B_0R_y(\theta)$. 

One of the main points of this article is to take seriously the fact that effects and preparations are indeed dual. We thus adopt the scheme in which we always start with a resource-free preparation $\omega_0\equiv \ket{00}\bra{00}$ (as both stabilizer and factorized) and perform resource-free measurements encoded in $B_0$. Then, all the resources are employed in the `unitary state evolution' $\omega_0\mapsto \omega_U\equiv U\omega_0 U^\dag$ to obtain the expectation value 
\ba
b_U := \Tr[B \omega_U] \equiv \Tr[B_U\omega_0].
\ea
This is the conceptual scheme of quantum computation in which one always initializes the state in all zeros, performs the computation by quantum operations, and finally measures in the computational basis \cite{Watrous_2018,kraus_book}. All the resources involved in the CHSH inequality have thus been encoded in the quantum operation (here unitary) $U$. This work shows what structure $U$ must have from the resource-theoretic point of view in order to violate the CHSH inequalities so all that follows can be shown for $B_0$ as resource-free operator without loss of generality. 

\section{Stabilizer Entropy\label{sec:SE}}

Consider a system of $n$ qubits and the set of Pauli strings $\mathbb{P}_n =\{{I},X,Y,Z\}^{\otimes n}$. The set $\mathbb{P}_n\times \{\pm 1,\pm i\}$ is the Pauli group $\mathcal P_n$ and its normallizer within the unitary group is the Clifford group $\mathcal C$. The orbit of the Clifford group through any computational basis state is the set of pure stabilizer states denoted by $\stab$. 
For a pure state $\ket{\psi}$ and $P\in \mathbb{P}_n$, the quantity 
$\Xi_P(\ket{\psi}) := d^{-1} (\Tr{P\psi})^2$ is a probability distribution over $\mathbb{P}_n$, with $d=2^n$. The $2$-SE $M_2(\ket{\psi})$ is defined as the (shifted) $2-$R\'enyi entropy
\ba\label{m2}
M_2(\ket{\psi}):= -\log\left[d\sum_{P\in \mathbb{P}_n} \Xi_P(\psi)^2\right]
=-\log \left( d^{-1}\sum_{P\in\mathbb{P}_n}\Tr^4 [\psi P] \right)
=-\log d\Tr[Q\dyad{\psi}^{\otimes4}],
\ea
where $Q:=d^{-2}\sum_{P\in\pauli{n}}P^{\otimes4}$.
The 2-SE can be extended to generic mixed states by $\tilde{M}_2(\rho)=M_2(\rho)-S_2(\rho)$, where  $S_2(\rho)=-\log\Tr[\rho^2]$ is the 2-R\'enyi entropy of $\rho$. Free states, denoted as $\stabzero$ are those of the form $\chi=d^{-1}\sum_{P\in\mathcal{G}}\phi_PP$, where $\mathcal{G}$ is an Abelian subgroup of the  Pauli group $\mathcal P_n\equiv \mathbb P_n \times \{\pm1, \pm i\}$ and $\phi_P=\pm1$ \cite{gottesman_heisenberg_1998}.  These are the (mixed) states that cannot be purified in a stabilizer state and one can show that $\stabzero = \{ \sigma \in \mathcal D(\mathcal H) : \tilde{M}_2 (\sigma)=0\}$ \cite{cao2024gravitationalbackreactionmagical}. The convex hull of $\stab$ is denoted by $\stabone \supset\stabzero\supset \stab$. 

The 2-SE  $\tilde{M}_2$ is a good monotone for pure states \cite{leone_stabilizer_2024},
faithful with respect to the free states, invariant under Clifford unitaries $C\in\mathcal C$ and additive under tensor product, i.e. $\tilde{M}_2(\rho\ot\sigma)=\tilde{M}_2(\rho)+\tilde{M}_2(\sigma)$.

The SE is a probe in the computational complexity of a quantum state. Indeed, stabilizer states can be simulated efficiently using only classical resources, and thus are useless in order to outperform classical computation. On the other hand, non stabilizer states, i.e. states having a non null stabilizer entropy, require exponential classical resources in order to be simulated. Moreover, the SE quantifies the distinguishability of a state from a Haar random state~\cite{bittel2026operationalinterpretationstabilizerentropy}, i.e. the larger the SE of a state, the harder it is to distinguish it from Haar random.

As in this work we want to characterize the operations $U$ from the resource-theoretic point of view, we  define the non-stabilizing power of a unitary $U$ as the average 2-SE created by the action of $U$ on the orbit of stabilizer states:
\ba
\mathcal{M}_2(U) :=\frac{1}{|\rm{STAB}|}\sum_{\ket{\psi}\in\rm{STAB}}M_2(U\ket{\psi}) \ . 
\ea
For example, the non-stabilizing power of $W(\theta)$ is
\ba
\mathcal{M}_2(W(\theta))=-\frac{4}{5} \log \left(\frac{7+\cos (4 \theta)}{8} \right).
\ea
In Fig.~\ref{fig:violations_theta} we see how the magic power of $W(\theta)$ follows closely the magic of  the state $W(\theta)|00\rangle$. Moreover the maximal CHSH violations coincide with the maxima of both $M_2$ and $\mathcal{M}_2$. 

Another  related quantity characterizes the interplay between SE and entanglement: the \textit{non-local non-stabilizerness} first introduced in \cite{cao2024gravitationalbackreactionmagical}. Given a bipartition $AB$ of the Hilbert space, the non-local non-stabilizerness $M_{\rm NL}$ is defined as:
\ba
\label{eq:non_local_magic_def}
M_{\rm NL}(\ket{\psi})=\min_{U_A\otimes U_B}M_2(U_A\ot U_B\ket{\psi}).
\ea
This quantity measures the amount of non-stabilizerness that is non-local, i.e.~that it cannot be erased from the state by means of local unitary operations. As the non-local SE is obtained from a minimization over local unitaries, it represents the amount of non-stabilizerness independent from the choice of the local basis. Another consequence of the minimization procedure is that non-local non-stabilizerness is solely dependent on the entanglement spectrum, and is thus considered to be representative of the interplay between entanglement and non-stabilizerness, which is the actual resource needed for quantum computation outperforming its classical counterpart. In \cite{qian2025quantumnonlocalmagic} an explicit expression for two qubit states is found: for any pure state $\ket{\psi}$ with entanglement spectrum $\{\cos^2(\theta),\sin^2(\theta)\}$, the non-local magic reads
\ba
    M_{\rm NL}(\ket\psi)=-\log\left(\frac{7+\cos(8\theta)}{8}\right) \; \ .
    \label{eq:m_nl_def}
\ea
The state such that $M_{\rm NL}(\ket\psi)=M_2(\ket\psi)$ is $\ket{r(\theta)}:=\cos\left(\theta\right)\ket{00}+\sin\left(\theta\right)\ket{11}$ with $\theta \in [0, \pi/2]$ modulo local Clifford unitaries. With these definitions in play, we are now ready to show the resource-theoretic structural elements of $U$ involved in violations of CHSH.

\section{Non-stabilizerness and violations of the CHSH inequality\label{sec:facts}}

In this section, we show some facts about the structure of entanglement and SE in the context of the CHSH inequality. Informally, we show that in order to violate the CHSH inequality  i) both entanglement and SE are necessary; ii) the preparation unitary $U$ must be asymmetric; iii) non-local magic hinders the violation of locality (!); and iv) probes of the interplay between entanglement and SE, like the capacity of entanglement, offer a valuable insight on the nature of the violation (or lack thereof).

Let us start by showing that $U$ must be both entangling and non-Clifford in order to violate CHSH. We indicate with $\mathcal{C}$ the Clifford group (the normalizer of the Pauli group). 

\begin{thm}
\label{theorem1}
Given an operator $B=P_A\ot(P_B+P_{B'})+P_{A'}\ot(P_B-P_{B'})$ with $P_{A,B,A',B'}\in\{X,Y,Z\}$, a state $\omega_0=\dyad{00}$ and a unitary Clifford operator $C\in\mathcal{C}$, then:
\ba
\label{eq:clifford_viol}
|\tr{BC\omega_0 C^\dag}|\leq 2.
\ea
Moreover, the same holds if the preparation $\omega_0$ is a mixed stabilizer state $\chi$ (obtained by convex combinations of pure stabilizer states):  $|\tr{B\chi}|\leq2$.
\end{thm}
\begin{proof}
We say that the operator $B$ is degenerate if at least two terms in $B$ are equal. One sees that if $B$ is degenerate, then it is of the form $B=2P_A \otimes P_B$ and the result is obvious since $\Vert B_0 \Vert = 2$ in this case and $\left \vert \Tr( B_0 \psi ) \right \vert \le \Vert B_0 \Vert$ for all states $\psi$. Hence, we can assume that $B$ is non-degenerate ($A\neq A'$ and $B\neq B'$). 
To prove the first statement of the theorem, let us first note that $\psi = C \omega_0 C^\dagger$ is a pure stabilizer state and so is of the form
\ba
\psi=\frac{1}{4}\sum_{P\in \mathcal{G}}\phi_{P}P,
\ea
where $\mathcal{G}$ is an abelian subgroup of the Pauli group of dimension four and $\phi_P =\pm 1$.
Now note that $B$ is a sum of four Pauli strings. Using orthogonality of Pauli strings $\tr{PP'}=d\delta_{PP'}$, we obtain
\ba
\label{eq:clifford_pure_dem}
b_C=\tr{BC\omega_0 C^\dag}=
\phi_{P_{AB}} +\phi_{P_{AB'}} + \phi_{P_{A'B}} - \phi_{P_{A'B'}} \ .
\ea
At this point note that there can be at most two Pauli strings in $B$ commuting with each other and thus belonging to the same Abelian subgroup. To see this, note that when $B$ is non-degenerate ($A\neq A'$ and $B\neq B'$), $P_A \otimes P_B$ and $P_{A'} \otimes P_{B'} $ always commute (different Pauli operators anticommute) and any attempt to enlarge this set makes $B$ degenerate.

As a consequence, at most two terms in Eq.~\eqref{eq:clifford_pure_dem} can be different from zero, from which the theorem follows.
As for the second part of the theorem, simply note that if $\chi = \sum_i p_i \chi_i$ with $\chi_i$ pure stabilizer states and probabilities $p_i$ summing to one, 

\begin{equation}  \left|\Tr[B\chi]\right|\le\sum_{i}p_{i}\left|\Tr[B\chi_{i}]\right|\le2.
\end{equation}

\end{proof}
\noindent Therefore, one cannot violate the CHSH inequality with either pure or mixed stabilizer states.

\begin{cor}
If the initial preparation is $\sigma \in\stabzero$, then $C\sigma C^\dag\in\stabzero$ and therefore  $\tilde{M}_2 (\psi)=0 \Rightarrow \left|\Tr[B\psi]\right|\le 2$.
\end{cor}

Let us now show that if we restrict ourselves to the class of operators $B\in\mathcal{B}$ that we call {\it symmetric}, for which $P_A = P_B$ and $P_{A'} = P_{B'}$ (of which $B_0$ is an example) and the unitary preparation $U$ is symmetric in $A$ and $B$, there cannot be a violation either. 
\begin{thm}
\label{theorem2}
Let ${\mathcal E}$ be a  CPTP map ${\mathcal E}: {\mathcal L} ({\mathcal H}_{AB}) \rightarrow {\mathcal L} ({\mathcal H}_{AB})$ commuting with the swap operator $T_2$, that is, ${\mathcal E} \left(T_2\rho T_2^\dagger\right) = T_2 {\mathcal E}(\rho) T_2^\dagger$; then, for the symmetric resource-free CHSH operator $B_0$, it holds
\ba
\left| \tr{(B_0{\mathcal E} (\omega_0)} \right|\le 2.
\ea
\begin{proof}
Since $[\omega_0,T_2]=0$, we have $T_2 {\mathcal E}(\omega_0) T_2^\dagger = {\mathcal E} (T_2 \omega_0 T_2^\dagger)=  {\mathcal E} (\omega_0)$. This state must be of the form ${\mathcal E} (\omega_0)=p_{\rm sym}\omega_{\rm sym}+p_{asym}\dyad{\psi^-}$, where $\omega_{\rm sym}$ is entirely supported on the symmetric subspace $1/2(I_A\otimes I_B + T_2) {\mathcal H}_{AB}$ and it is therefore of the form  
\ba
\omega_{\rm sym} = \sum_i p_i \ket{\phi_i}\bra{\phi_i},
\ea
with $ \ket\phi_i=a_i\ket{00}+b_i\ket{11}+c_i\ket{\psi^+}$, $\ket{\psi^+}=\ket{01}+\ket{10}/{\sqrt{2}}$ and $|a|^2+|b|^2+|c|^2=1$. The anti-symmetric subspace is instead spanned by the state $\ket{\psi^-}=\ket{01}-\ket{10}/\sqrt{2}$. By triangular inequality we obtain $|\tr{B_0\mathcal{E}(\omega_0)}|\leq p_{\rm sym}|\tr{B_0}\omega_{sym}|+p_{\rm asym}|\tr{B_0\dyad{\psi^-}}|$. Thus, in order to prove our statement we need to bound both expectation values, the one on the symmetric subspace, and the one on the anti-symmetric subspace. For the symmetric subspace, exploiting again the triangular inequality, we have $\left| \tr{(B_0\omega_{\rm sym}} \right|\le \sum_i p_i | \tr{B_0 \dyad{\phi_i}}|$. We need therefore to bound this expectation value for an arbitrary pure state in the symmetric subspace. Direct evaluation of the expectation value of $B$ for symmetric resource free CHSH operators yields
\begin{equation}
    \left|\tr{B U_{\rm sym}\ket{00}\bra{00}U_{\rm sym}^\dag}\right|=\Big \{ 2 |c|^2, |a+b|^2, |a-b|^2\Big \} \leq 2 \, , 
\end{equation}
where we used $ \Vert \mathbf{x} \Vert_1  \le \sqrt{2} \Vert \mathbf{x} \Vert_2$ for $\mathbf{x}=(a,b)$.

On the other hand, the state $\ket{\psi^-}$ is a stabilizer state, so that by Theorem~\ref{theorem1}, its expectation value is also bounded by 2. This concludes the proof.
\end{proof}
\end{thm}
\noindent
The above theorems show that Bob and Alice must measure in directions that are tilted with respect to each other, and that this tilt must be a non Clifford rotation.

Finally, we show  structural relations between non-locality and  non-stabilizerness in the CHSH violations.
Since violations of the CHSH are connected with non-local behavior, one would naively expect non-local magic to play a major role in CHSH violations. It turns out that not only this is not the case, but it is actually the opposite: non-local magic is detrimental for the maximal violation of CHSH inequality. Indeed, one can prove that non-local non-stabilizerness bounds the maximal achievable violation, i.e. the larger $M_{\rm NL}$ the lower the maximal violation. Moreover, we are going to see that in order to violate the CHSH inequality, one needs the operator $U$ to inject some \textit{local} non-stabilizerness, defined as the difference between the total $M_2$ and the non-local one:
\begin{equation}
    M_{\rm LOC} (|\psi\rangle ) := M_2 (|\psi\rangle ) - M_\mathrm{NL}(|\psi\rangle ) \ge 0. 
\end{equation}
These results provide an alternative view, in terms of quantum resources, of the necessary (known) conditions needed to violate the CHSH inequality: while in the usual view one assigns non-local properties to the state and local non-stabilizerness to the measurements, in this setting all resources are assigned to the state, and characterized accordingly.
Let us now prove that $M_{\rm NL}$ bounds the maximal violation.

\begin{thm}\label{TheoremI}
For a general $2-$qubit state $\ket{\psi(\theta)}$ with Schmidt coefficients $\cos\left(\theta\right), \sin\left(\theta\right)$ and the Bell operator $B=P_A\ot(P_B+P_{B'})+P_{A'}\ot(P_B-P_{B'})$ with $P_{A,B,A',B'}\in\{X,Y,Z\}$, it holds:
\ba
 C(\theta) :=\max_{U_A \otimes U_B}  \left| \tr{B_{U_A \otimes U_B}\dyad{r(\theta)}}\right| \leq 2 \sqrt{2} - \frac{\ln(2)}{\sqrt{2}} M_{NL}(\theta).
    \label{eq:ineq_B_M_NL}
\ea
where $B_{U_A\ot U_B}\equiv(U_A^\dag\ot U_B^\dag)B(U_A\ot U_B)$ and $M_{\rm NL}(\theta)=\min_{U_A\ot U_B} M_2(U_A\ot U_B\ket{r(\theta)})$.
\begin{proof}
While the Bell operator $B$ is defined in terms of specific Pauli measurements, the maximization over local unitaries $U_A \otimes U_B$ ensures that $C(\theta)$ corresponds to the maximal violation of the CHSH inequality allowed by the state's correlation matrix. By the Horodecki criterion \cite{Horodecki1995}, this value is a function solely of the Schmidt coefficients, allowing for a tight comparison with the local-unitary invariant $M_{NL}(\theta)$. In particular:
\begin{align}
    C(\theta) = 2\sqrt{1+\sin^2(2 \theta)}.
\end{align}
By defining the parameter $t:=\sin^2(2\theta) \in [0,1]$, we have:
\begin{align}
    C(t) = 2 \sqrt{1+t},\qquad M_{NL} (t) = -\log\left( 1-t+t^2\right),
\end{align}
where the expression for $M_{NL}$ follows from Eq.~\eqref{eq:m_nl_def}.
To determine the optimal coefficient $\alpha$ such that the inequality: 
\begin{align} \label{Eq. C(t)Inequality}
    C(t) \leq 2\sqrt{2} - \alpha M_{NL}(t),
\end{align}
holds for all $t$, we define the ratio:
\begin{align} \label{Eq. C(t)InequalityEquiv}
    \alpha \geq \frac{2\sqrt{1+t}-2\sqrt{2}}{\log(1-t+t^2)} =: g(t).
\end{align}
We investigate the asymptotic behavior of $g(t)$ in the limit of maximal entanglement ($t \rightarrow 1$). By defining $t=1-\epsilon$ and applying a Taylor expansion for the logarithmic and square-root terms~(see App.~\ref{appendixxx}), we find:   
\begin{align}
    \alpha_{\text{opt}} = \lim_{t \rightarrow 1} g(t) =\frac{\ln(2)}{\sqrt{2}}.
\end{align}
The sufficiency of the optimal coefficient, ensuring that the inequality remains satisfied across the entire domain $t \in [0,1]$, is shown in  App.~\ref{appendixxx} via a monotonicity study of the associated gap function, alongside the complete algebraic derivation of the functional forms.  
\end{proof}
\end{thm}

Theorem~\ref{TheoremI} bears several consequences. First, as expressed by the following corollary, the presence of the smallest amount of non-local non-stabilizerness hinders the possibility of achieving the Tsirelson bound.

\begin{cor}
\label{thm:no_tsirelson}
Given a state $\psi$ such that ${M}_{\rm NL}\neq0$, then
\ba
\left|\tr{B\psi}\right|<2\sqrt{2}.
\ea
\begin{proof}
It's a trivial consequence of Theorem~\ref{TheoremI}.
\end{proof}
\end{cor}
\noindent Thus, injecting non-local non-stabilizerness is actually detrimental for non-local violations. Moreover, we also provide the optimal coefficient to bound the maximal expectation value of the Bell operator $B$ in terms of non-local non-stabilizerness.

To conclude, we now prove that \textit{local} non-stabilizerness is necessary in order to violate the CHSH inequality.

\begin{thm}
\label{theorem4}
If a unitary operator $U$ does not inject any local magic,  that is $M_{\rm LOC}(U\ket{00})=0$, then
\ba
|b_U|=\left|\tr{B U\omega_0U^\dag}\right|\leq2.
\ea
\begin{proof}
If there is zero local magic, the state $U\ket{00}$ has the form \cite{qian2025quantumnonlocalmagic}
\ba
\label{eq:r_theta}
U\ket{00}& =& C_A\otimes C_B ({\cos\left(\theta\right)}\ket{00} + {\sin\left(\theta\right)}\ket{11} ) \equiv C_A\otimes C_B (\ket{r(\theta)}) 
\\
&=& \cos\left(\theta\right) \ket{ss'} +\sin\left(\theta\right) \ket{\bar{s}\bar{s}'},
\ea
where $\{\ket{s}, \ket{\bar{s}} \}$ is a basis for Alice of eigenstates of Pauli operators and similarly for Bob, and $C_A,C_B$ are local Clifford unitaries, which map $\ket{00}$ and $\ket{11}$ into the tensor product of eigenstates of other single qubit Pauli operators. Since $B=C_{A'}^\dagger \otimes C_{B'}^\dagger B_0 C_{A'} \otimes C_{B'}$, it is sufficient to check the statement for $B_0$. One can then directly verify that for all possible combination of eigenstates of $\{X,Y,Z\}$ the inequality holds, proving the theorem.
\end{proof}
\end{thm}
\subsection{Probes of magic and non-locality}

Let us now show an example of the  very rich interplay between entanglement $E_{\mathrm{VN}}(|\psi\rangle ) = S_1 (\psi_A)$, magic $M_{2}(\ket\psi)$, non-local magic $M_{\rm NL}(|\psi\rangle )$,  the possibility of violation of CHSH inequality, and its maximum entity. Remarkably, non-local magic takes into account both entanglement and non-stabilizerness as factorized states have obviously zero non-local magic. As evaluating non-local magic is, in general, a daunting task, we also study the capacity of entanglement $C_E$ as a quantity that has some properties in common with it and can serve as a probe. 
$C_E(|\psi\rangle)$ is defined as
\begin{align}
    C_E(|\psi\rangle ) :=\left\langle(\log \psi_A)^2\right\rangle_{\psi_A}-\left\langle\log \psi_A\right\rangle_{\psi_A}^2 
    = \Tr\left[\psi_{A}\left(\log\psi_{A}\right)^{2}\right]-\left(\Tr[\psi_{A}\log\psi_{A}]\right)^{2} \, .
\end{align}
The entanglement capacity is a measure of how much the reduced state is {\it non-flat}, i.e.~it measures how much $\psi_A$ deviates from being proportional to a projector~\cite{PRXQuantum.3.010325,8765829,7001656}.
Moreover, the entanglement capacity, in contrast with the most common entanglement measures, such as the entanglement entropy or the concurrence, is zero not only for separable states, but also for maximally entangled states. Moreover, the maximum of the entanglement capacity is reached for partially entangled states~\cite{Wei_2023,nandy_capacity_2021}. The entanglement capacity is also a measure of the fluctuations of the entanglement Hamiltonian~\cite{okuyama_capacity_2021}. Its connection with $M_{NL}$ is in the fact that $C_E(\psi)=0$ iff $\psi$ has zero non-local magic~\cite{cao2024gravitationalbackreactionmagical}.
The $C_E$  has found applications in many-body systems, connecting thermodynamic quantities and the Rényi entropies~\cite{PhysRevLett.105.080501,PhysRevB.83.115322}, and the AdS/CFT correspondence~\cite{Dong2016,PhysRevLett.122.041602,cao2024gravitationalbackreactionmagical}, where it has a relatively simple bulk interpretation given by metric fluctuations integrated over the Ryu-Takayanagi surface, i.e. the entangling surface. 

Let us start with the family of states defined by 
\ba
  \ket{\rho} &=& \sqrt{\tfrac{1+r}{2}}\,\cos\!\big(\tfrac{\theta}{2}\big)\,\ket{00}
   +\sqrt{\tfrac{1-r}{2}}\,\sin\!\big(\tfrac{\theta}{2}\big)\,\ket{01} \nonumber\\
   &+&\sqrt{\tfrac{1+r}{2}}\,e^{i\phi}\sin\!\big(\tfrac{\theta}{2}\big)\,\ket{10}
   -\sqrt{\tfrac{1-r}{2}}\,e^{i\phi}\cos\!\big(\tfrac{\theta}{2}\big)\,\ket{11}\\
   &=&    \sqrt{\tfrac{1+r}{2}}\,\ket{\rho_+}_A\ket{0}_B
   +\sqrt{\tfrac{1-r}{2}}\,\ket{\rho_-}_A\ket{1}_B
\ea
with $r\in[0,1]$, $\theta \in[0,\pi]$, $\phi \in [0, 2 \pi)$ and 
\ba
    \ket{\rho_+}&= \cos\left({\frac{\theta}{2}}\right) \ket{0}+ e^{i \phi} \sin\left({\frac{\theta}{2}}\right)\ket{1}\, ,\\
    \ket{\rho_-}&= \sin\left({\frac{\theta}{2}}\right)\ket{0}- e^{i \phi} \cos\left({\frac{\theta}{2}}\right) \ket{1}\, .
\ea
Setting $\theta=\pi/4$ and $\phi=\pi/3$, we obtain
\begin{equation}
   \rho_A=\Tr_B\dyad{\rho}
   =\frac{I+\Vec{r}\cdot\Vec{\sigma}}{2}
   =\frac{1+r}{2}\dyad{\rho_+}+\frac{1-r}{2}\dyad{\rho_-},
\end{equation}
The evaluation of the figures of merit gives
\ba
S_1(\rho_A) &=& 1-\frac{1}{2}\Big[(1+r)\log(1+r)+(1-r)\log(1-r)\Big] \\
C_E(\rho_A)&=&\frac{1-r^2}{4}\,\log^2\!\Big(\frac{1+r}{1-r}\Big)\\
   M^{NL}({\rho})&=&-\log\!\big(1-r^2+r^4\big)\\
 \tr{B_0 \dyad{\rho}}  &=& \frac{3}{8}\Big[(2+\sqrt2)\sqrt{1-r^2}+2\sqrt2\Big]
\ea
\begin{figure}[!t]
    \centering
    \includegraphics[width=0.7\linewidth]{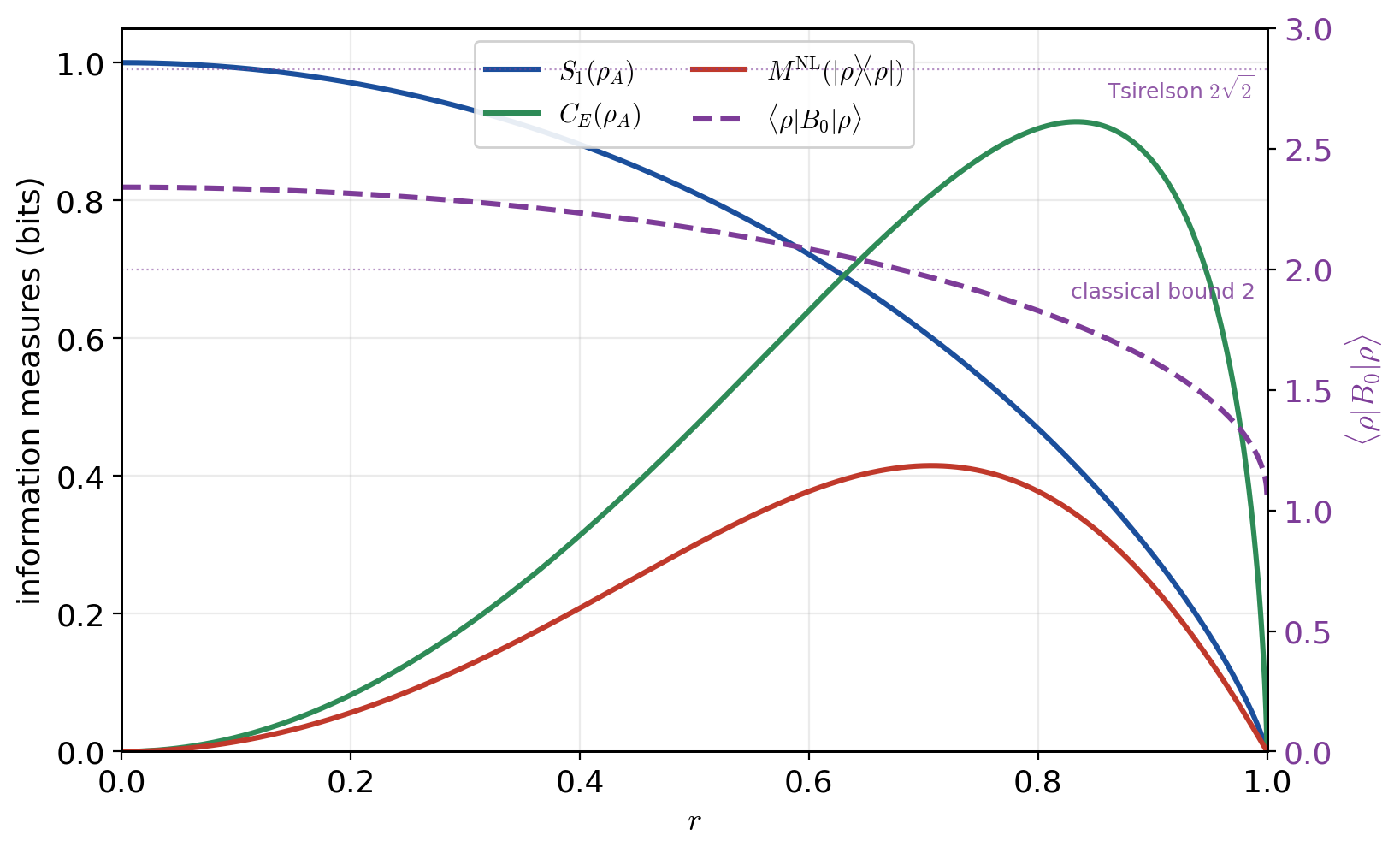}
    \caption{ Plot of the entanglement entropy $S_1(\rho_A)$ (magenta solid line), the capacity of entanglement $C_E(\rho_A)$ (dashed indigo line), the non-local non-stabilizerness $M_{NL}(\rho)$ (dotted orange line) and of $\Tr[B_0\dyad{\rho}]$ as a function of $r$, having set $\theta=\pi/4$ and $\phi=\pi/3$ (dotted-dashed green line). For $r=0$ the corresponding state is a combination of all four computational basis states, $M_{\rm NL}=0$ as well as $C_E$, while entanglement is maximum. Under these conditions, the CHSH inequality (dashed red line) is violated. One can then observe the decrease of the CHSH violation as $C_E$ and $M_{\rm NL}$ grow, up to a point where no violation is observed, in spite of the state being still highly entangled, as proven by the non-zero value of $S_1$. Finally, for $r=1$, both $M_{\rm NL}$ and $C_E$ are once again zero, but it is still not possible to violate the CHSH inequality because of the missing entanglement, since $S_1=0$.}
    \label{fig:allquantities}
\end{figure}

The results are summarized in Fig.~\ref{fig:allquantities}. First of all, we see how the qualitative similar trend (and same values at the boundaries) makes $C_E$ a good probe for $M_{\rm NL}$. Unfortunately, for two qubits they are not strictly monotone with each other. In this family one has maximal CHSH violation for  $r=0$  while the non-local magic and the entanglement capacity are both zero. As $r$ grows, the expectation value of the Bell operator $B_0$ decreases, while the non-local magic increases, showing how non-local non-stabilizerness may hinder the violation of the CHSH inequality. Beyond a critical value of $r$, CHSH violations are not observed anymore, as there is too much non-local magic. Notice also that at the critical value of $r$ for which violations are not observed anymore, the state is still entangled, as shown by the entanglement entropy.

\subsection{CHSH geometry}
Let us now try to get some geometric understanding regarding  the region of pure states that violate the CHSH inequality. We start by going to the eigenbasis of the operator
$
B_{0}=\sum_{i}\lambda_{i}|\phi_{i}\rangle\langle\phi_{i}|,  
$
where we ordered the eigenvalues as $\lambda_{i}=\left\{ -2\sqrt{2},2\sqrt{2},0,0\right\} $.
A generic pure state $|\psi\rangle$ in this eigenbasis reads $|\psi\rangle=\sum_{i}\psi_{i}|\phi_{i}\rangle$.
The condition $\left|\langle\psi|B_{0}|\psi\rangle\right|>2$ does not identify a vector subspace. 
Indeed, it can be rewritten as:
\ba
\label{eq:viol_cond}
\left|\langle\psi|\left(\sum_{i}\lambda_{i}|\phi_{i}\rangle\langle\phi_{i}|\right)|\psi\rangle\right|  >2
\Leftrightarrow\left|\left(\sum_{i}\lambda_{i}\left|\psi_{i}\right|^{2}\right)\right|  >2
\Leftrightarrow\,\left|\left|\langle\phi_{1}|\psi\rangle\right|^{2}-\left|\langle\phi_{2}|\psi\rangle\right|^{2}\right|>\frac{1}{\sqrt{2}}.
\ea
Let us now  write the state $\ket{\psi}$ in this basis according to the Hurwitz parametrization~\cite{Todd_Tilma_2002} of a general pure state:
\begin{align}
\left\{ \psi_{i}\right\} _{i=1}^{4} = &\Big(\cos\left(\vartheta_{3}\right),\sin\left(\vartheta_{3}\right)\cos\left(\vartheta_{2}\right)e^{i\phi_{3}}, \nonumber \\ & \ \sin\left(\vartheta_{3}\right)\sin\left(\vartheta_{2}\right)\cos\left(\vartheta_{1}\right)e^{i\phi_{2}},
\sin\left(\vartheta_{3}\right)\sin\left(\vartheta_{2}\right)\sin\left(\vartheta_{1}\right)e^{i\phi_{1}}\Big),
\end{align}
where $\vartheta_{i}\in\left[0,\pi/2\right]$ and $\phi_{i}\in[0,2\pi]$ are the six parameters describing a two qubit pure state. Inserting this parametrization into Eq.~\eqref{eq:viol_cond} one obtains:
\ba
\left|\cos^{2}\left(\vartheta_{3}\right)-\sin^{2}\left(\vartheta_{3}\right)\cos^{2}\left(\vartheta_{2}\right)\right|>\frac{1}{\sqrt{2}}.   
\ea
Thus, in the end only two parameters enter the violation of the CHSH inequality, allowing for a graphical representation, as shown in Fig.~\ref{fig:Bell_violation}, where the density of states violating the CHSH inequality is shown in the plane $\{\theta_2,\theta_3\}$. The main message of this short digression is that the region of pure states violating the CHSH inequality is non-trivial, and moreover the region of maximal violation has very little weight as can be seen from Fig.~\ref{fig:Bell_violation}. 

\begin{figure}[!t]
\begin{subfigure}[t]{0.44\linewidth}
\centering
\includegraphics[width=\linewidth]{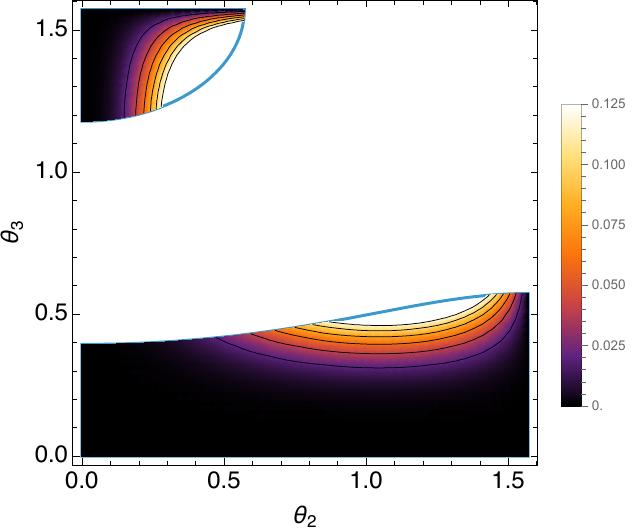}
\end{subfigure}
\begin{subfigure}[t]{0.44\linewidth}
\centering
\includegraphics[width=\linewidth]{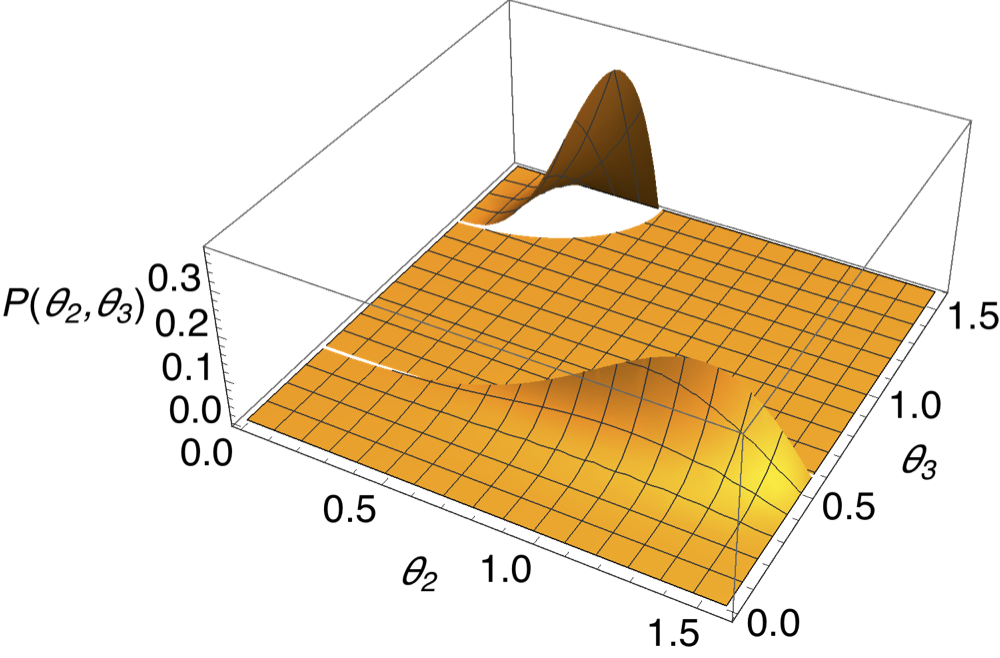}
\end{subfigure}
\caption{Density of states violating the CHSH inequality in the $\{\theta_2,\theta_3\}$ plane.
The region of violation of the Bell's inequality. Left panel: density plot, the color code corresponds
to the density of states in this chart. Right panel: 3D plot. Note that in these coordinates the maximal violation corresponds to the point $(\theta_2,\theta_3)=(0,\pi/2)$ and the line $(\theta_2,\theta_3)=(\theta,0)$ for $\theta\in [0,\pi/2]$. \label{fig:Bell_violation} }
\end{figure}
\section{Random non-locality\label{sec:random}}
In this section, we analyze the probability of violating the CHSH inequality when the state-preparing unitary $U$ is taken from an ensemble $\mathcal E_U$ of unitaries with respect to a measure $d\mu_U$. The ensemble $\mathcal E_U$ represents a lack of control on the preparation unitary $U$. Then we ask which ensembles are more likely to provide a violation. The choice of the ensembles $\mathcal E_U$ is of course in principle experimentally motivated, however, within the same experimental capabilities one could have access to different ensembles $\mathcal E_U$. The theorems and facts of the previous section have shown that certain ensembles of unitaries would be useless: obviously, the ensemble of factorized unitaries $\mathcal E_{U_A\otimes U_B}$, the set of Clifford unitaries $\Cl$, but also the symmetric unitaries $U_\mathrm{sym}$ considered in Theorem \ref{theorem2} and the non-local unitaries that do not produce local magic $\mathcal{E}_{U_{NL}}:=\{ U |\; M_\mathrm{LOC}(\omega_U)= 0\}$. This suggests that we can improve the chances of violating the CHSH inequality making use of the structure of $U$.

\subsection{CHSH in the Hilbert space}
We begin by considering $\mathcal E_U$ as the full unitary group endowed with the uniform (Haar) measure $d\mu_H$. In this extreme case, one has zero control whatsoever on the state preparation. We are asking what is the likelihood of violating the CHSH inequality for a completely random unitary $U$. First, we can obtain a rough estimate of this probability using Chebyshev's inequality, once the mean and standard deviation of the distribution of $b_U$ are known.  For the mean, using standard techniques~\cite{Mele2024introductiontohaar}, we obtain
\ba
\langle b_U\rangle_U=\left\langle\tr{B_0U\omega_0U^\dag}\right\rangle_U=\int_{\cal U}\,d\mu_H\tr{B_0U\omega_0U^\dag}
=\frac{1}{4}\tr{B_0}=0 \, .
\ea
Thus, on average, using a random uniformly distributed unitary $U$, one obtains zero as a result of the CHSH experiment. With the same techniques, one obtains for the variance $\Var_U(b_U)=4/5$ and thus, by the Chebyshev inequality, 
$\mathrm{Prob}(|b_U|>2)\leq 1/5$.
As we shall  see, this upper bound is very loose.  Indeed, in case of the full unitary group equipped with the Haar measure, it is possible to compute  the
probability of violating the CHSH inequality exactly using the results in  \cite{campos_venuti_probability_2013}. 

To this end, we need the probability distribution of obtaining a given outcome $x$ in the CHSH experiment, that is:
\ba\label{allisbU}
P_{B_0}(x) := \left\langle \delta(b_U-x)\right\rangle_\mathcal{U} \, . 
\ea
 The probability of violation is simply obtained integrating this probability distribution over the values corresponding to a violation, that is:
\ba
P_{\mathrm{viol}} & =\int_{|x|>2}dx P_{B_0}(x).
\label{eq:integral_violation}
\ea
Notice that, as shown in \cite{campos_venuti_probability_2013},  $P_{B_{0}}(x)$ is entirely determined by the spectrum of $B_0$, $\left\{ -2\sqrt{2},0,0,2\sqrt{2}\right\} $, including degeneracies. Hence, the result is the same for all the non-degenerate CHSH operators in $\mathcal{B}$ as expected (as it is well known that they are isospectral). 
To obtain $P_{B_{0}}(x)$ one can use the explicit formula Eq.~(25) in \cite{campos_venuti_probability_2013} for degenerate eigenvalues or lift the degeneracy of the zero eigenvalue to $-\epsilon,+\epsilon$, use the more manageable Eq.~(17) in \cite{campos_venuti_probability_2013}, and send $\epsilon\to0$ at the end. The result is
\begin{align}
\nonumber
P_{B_0}(x)  &=\frac{3\left(2\sqrt{2}+x\right)^{2}\sign\left(2\sqrt{2}+x\right)}{64\sqrt{2}}+\frac{3\left(2\sqrt{2}-x\right)^{2}\sign\left(2\sqrt{2}-x\right)}{64\sqrt{2}}-\frac{3}{8}|x| \\
& = \frac{3}{64}\left(8\sqrt{2}+\left|x\right|\left(\sqrt{2}\left|x\right|-8\right)\right)\1_{[-\sqrt{8},\sqrt{8}]}(x), 
\end{align}
where $\1_{A}(x)$ is the characteristic function of the set $A$. 
Computing the integral in Eq.~(\ref{eq:integral_violation}) we arrive to the following.
\begin{proposizione}
Given a Haar random unitary $U$, the probability of the state $\ket{\psi}=U\ket{00}$ to violate the CHSH inequality is:
\ba
P_{\mathrm{viol}}=\frac{\left(10-7\sqrt{2}\right)}{4}\approx0.0251=2.51\%.
\ea
\end{proposizione}
Note that the estimate obtained using the Chebyshev's bound is almost ten times larger than the actual probability.

\begin{figure}[!t]
    \centering
    \begin{subfigure}[t]{0.44\linewidth}
        \centering
        \includegraphics[width=\linewidth]{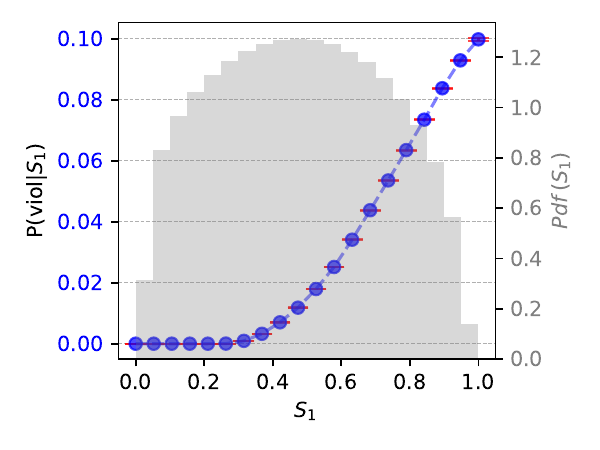}
        \caption{Probability of violating the CHSH inequality vs the entanglement entropy $S_1(\rho_A)$.}
        \label{fig:violations_entanglement}
    \end{subfigure}
    \begin{subfigure}[t]{0.44\linewidth}
        \centering
        \includegraphics[width=\linewidth]{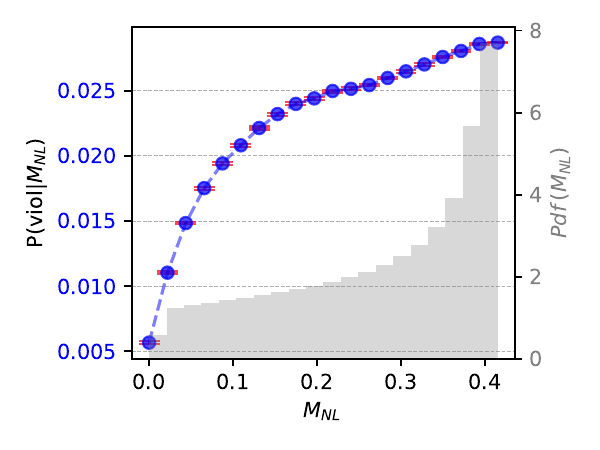}
        \caption{Probability of violating the CHSH inequality vs the non-local non-stabilizerness $M_{\rm NL}(\ket{\psi})$.}
        \label{fig:violations_m_nl}
    \end{subfigure}

    \begin{subfigure}[t]{0.45\linewidth}
        \centering
        \includegraphics[width=\linewidth]{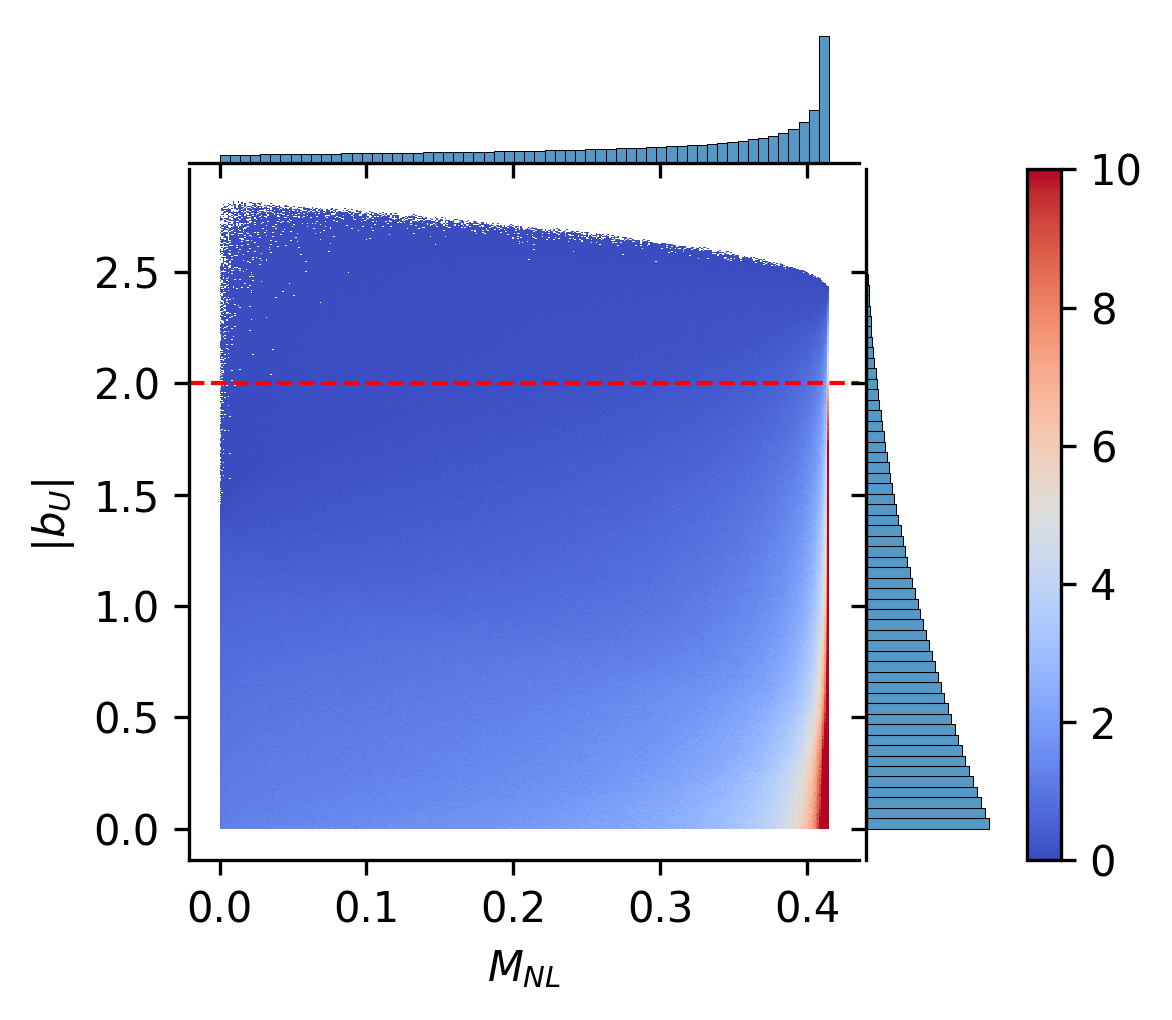}
        \caption{Density of states in the $|b_U|, M_{\rm NL}$ plane. Darker regions correspond to a higher density. On the axes, the marginal probability distributions of the respective quantities.}
        \label{fig:joint_MNL_Bell_exp_val}
    \end{subfigure}
    \begin{subfigure}[t]{0.44\linewidth}
        \centering
        \includegraphics[width=\linewidth]{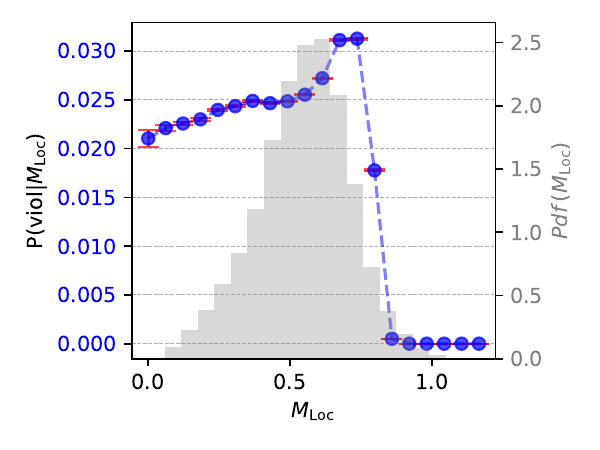}
        \caption{Probability of violating the CHSH inequality vs the local non-stabilizerness $M_{\rm LOC}(\ket{\psi})$.}
        \label{fig:violations_ml}
    \end{subfigure}
    
    \caption{We sample \(5 \times 10^7\) Haar-random two-qubit states and plot in blue the probability of violating the CHSH inequality as a function of (\ref{fig:violations_entanglement}) the entanglement entropy \(S_1\), (\ref{fig:violations_m_nl}) non-local non-stabilizerness \(M_{NL}\), and (\ref{fig:violations_ml}) local non-stabilizerness. Red error bars indicate the standard error in each bin, while the marginal probability distribution of the respective resources is shown in grey.
    In (\ref{fig:joint_MNL_Bell_exp_val}), we plot the joint distribution of $|b_U|$ and $M_{NL}$. Here, the maximum of $|b_U|$ is seen to be monotonously decreasing in $M_{NL}$ and the density of states justifies the results of Fig.~\ref{fig:violations_m_nl}, see discussion in main text.
    }
    \label{fig:violations_all}
\end{figure}

\begin{figure}[!t]
    \centering
    \begin{subfigure}[t]{0.43\linewidth}
        \centering
        \includegraphics[width=\linewidth]{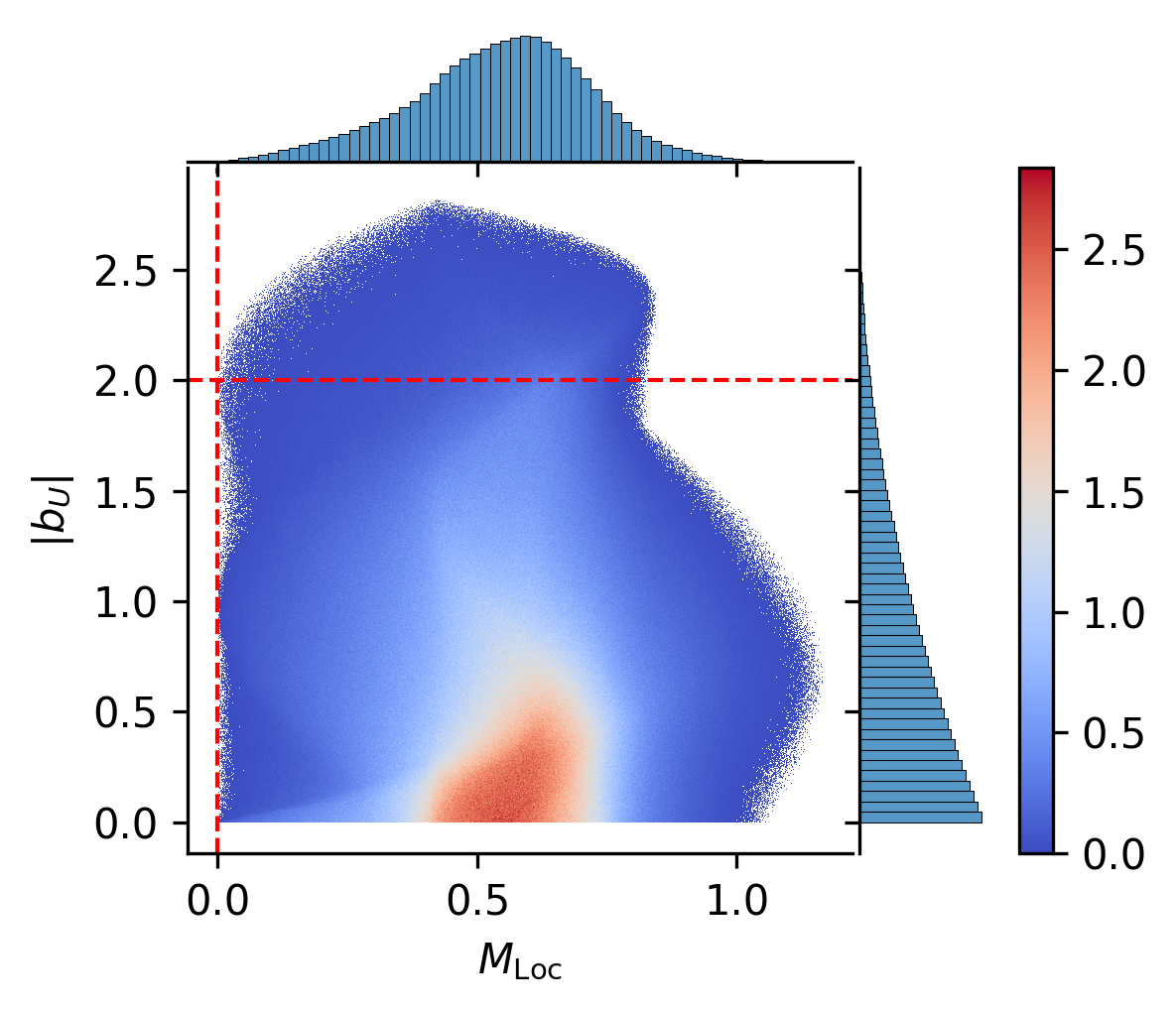}
        \caption{Density of states in the $|b_U|, M_\mathrm{LOC}$ plane. Darker regions correspond to a higher density. On the axes, the marginal probability distributions of the respective quantities. }
        \label{fig:density_states}
    \end{subfigure}
    \begin{subfigure}[t]{0.4\linewidth}
        \centering
        \includegraphics[width=\linewidth]{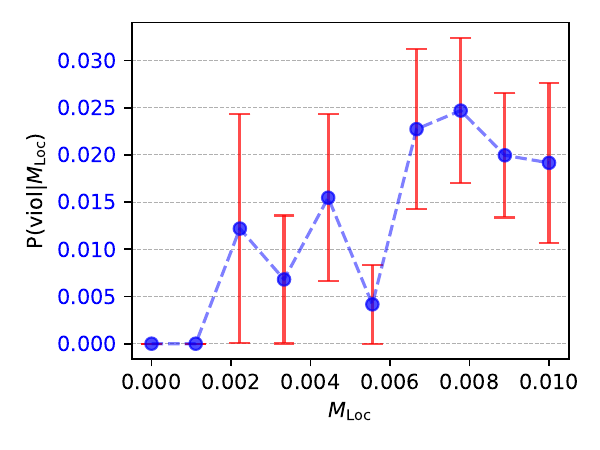}
        \caption{Fine-grained binning near the origin for the conditional probability of violation given $M_\mathrm{LOC}$.}
        \label{fig:violations_m_l_small}
    \end{subfigure}
    
    \caption{To understand the behavior of $P(\mathrm{viol}|M_\mathrm{LOC})$ for small $\mathrm{L}$, in  Fig.~\ref{fig:density_states}  we plot the density of states in the $|b_U|, M_{\rm NL}$ plane, and in Fig.~\ref{fig:violations_m_l_small} the probability of violations for states with small values of $M_{\rm LOC}$. One can observe how states with high values of $M_{\rm LOC}$ do not violate the CHSH inequality, because of the contemporary presence of non-local non-stabilizerness. In Fig.~\ref{fig:violations_m_l_small} one can see that the probability of violations tends to zero when $M_{\rm LOC} \to 0$.
    }
    \label{fig:ultima}
\end{figure}

To obtain a more detailed understanding, we analyze numerically the probability of violating the CHSH inequality given a fixed amount of resources contained in the state, let them be entanglement or non-stabilizerness. 
In practice, we compute numerically  the conditional probabilities\footnote{In general, to avoid situations like the Kolmogorov-Borel paradox, the conditional probability for continuous variables, corresponding to events with probability zero, must be defined as a limiting procedure. This problem does not arise in our discretized numerical simulations. On the contrary Eq.~(\ref{eq:cond_prob}) can be seen as the limit of our numerics when the number of samples $\to \infty$ and the size of the bins $\to 0 $. }
\begin{equation}
    P_{B|Y} (b,y) = \frac{\langle \delta (b_U - b) \delta(Y_U-y) \rangle_\mathcal{U}}{P_Y(y)} \, , \label{eq:cond_prob} 
\end{equation}
where $Y$ is a given resource, and $P_Y(y)$ its  density. 
The probability of violation given fixed resources $Y=y$ is then
\be
\mathrm{Prob}(\mathrm{violation} |Y=y ) = \int_{|b|>2} db \,  P_{B|Y} (b,y) \, .
\ee
We perform our analysis for $Y$ equal to the entanglement entropy $E_\mathrm{VN}$, the non-local magic $M_{\rm NL}$, and the local non-stabilizerness $M_{\rm LOC}$. The results are shown in Fig.~\ref{fig:violations_all}.

Fig.~\ref{fig:violations_entanglement} expresses the known fact that one needs a finite amount of entanglement to violate CHSH inequality.
In Fig.~\ref{fig:violations_m_nl} we show the conditional probability of violating the CHSH inequality at given values of the non-local non-stabilizerness. The results shown in this plot, lead us to observe the following fact.
\begin{fact}
\label{fact:non_local_probability}
States with large amounts of non-local non-stabilizerness are likely to violate the CHSH inequality.
\end{fact}
A naive interpretation of ~\ref{fact:non_local_probability} would lead one to think of a contradiction with Theorem~\ref{thm:no_tsirelson}, since the probability of violation increases with $M_{\rm NL}$. However, this can be explained by the fact that local and non-local magic of Haar-random states are not independent: states with high amounts of non-local magic will also possess local magic, and so the probability of violation increases. However, because of the presence of non-local magic, the violations are small, i.e.~$|b_U|$ is slightly above 2 and way below the Tsirelson bound. This interpretation is strengthened by the plot in Fig.~\ref{fig:joint_MNL_Bell_exp_val} where we show the density of states in the $|b_U|, M_{\rm NL}$ plane, leading to the next fact.
\begin{fact}
For small values of $M_{\rm NL}$ there is a low density of states violating the CHSH inequality. These can reach higher values of the violation. As $M_{\rm NL}$ increases, the maximum value of $|b_U|$ decreases, in accordance with Eq.~(\ref{eq:ineq_B_M_NL}).
\end{fact}

Finally, in Fig.~\ref{fig:violations_ml} we plot the probability of violations given $M_\mathrm{LOC}$. One can observe that the probability is non monotone with respect to local non-stabilizerness confirming the non-trivial interplay between local and non-local magic, since states with high amounts of local non-stabilizerness are constrained to have also large amounts of non-local non-stabilizerness, hindering the possibility of non-local violations. Note that, according to Theorem~\ref{theorem4}, $P_\mathrm{viol}=0$ when $M_\mathrm{LOC}=0$. The behavior of $P(\mathrm{viol}| M_\mathrm{LOC})$ for small $M_\mathrm{LOC}$ is detailed in Fig.~\ref{fig:ultima}  where it is confirmed that $P_\mathrm{viol}\to 0$ when $M_\mathrm{LOC} \to 0$.

\subsection{Isospectral twirling and ensembles of unitary operators}

Here we provide a systematic way to obtain a useful heuristic for the CHSH violation with limited control. The strategy is the following: we compute analytically the first two moments of the distribution of $b_U$ given by $d\mu_U$. Using the Chebyshev inequality, we argue about the most promising ensembles, i.e.~those that, according to the inequality, give the largest probability of violation. Then, numerically, we verify if the promising ensembles do indeed (mostly) provide a better likelihood for a violation. To this end, we will employ the technique of isospectral twirling~\cite{10.21468/SciPostPhys.10.3.076, leone_isospectral_2021,cusumano2026probeschaoscliffordgroup} that has been developed to model situations where one has good control over the eigenvalues of a unitary $U$ but limited control over its eigenstates.
We use this technique to construct ensembles $\mathcal E$ of unitaries that provide higher probability of violating the CHSH inequalities. We will construct the ensembles by utilizing insights given by structural properties of $U$ informed by the theorems of the previous sections.
We first define the ensemble $\mathcal E\equiv 
\{ g^\dag U_c g | g\in \mathcal{G} \}$ associated to a core unitary $U_\mathrm{c}$ and  $\mathcal G\subseteq\mathcal U(d)$ being a subgroup of the full unitary group $\mathcal U(d)$. From now on, the unitary $U_\mathrm{c}$ fixing the spectrum will be called the {\it core} and the group $\mathcal{G}$ will be explicitly denoted in the average operation $\langle \cdot\rangle_{\mathcal G}$. Given the core operator $U_\mathrm{c}$, the ensemble $\mathcal{E}$ consists of operators with the same eigenvalues as $U_\mathrm{c}$ but with eigenvectors determined by the action of $\mathcal{G}$ on $U_\mathrm{c}$.  One can think that isospectral twirling mimics a situation where an experimenter tries to prepare the core unitary $U_\mathrm{c}$ and achieves in preparing the ensemble $\mathcal{E}$ due to the effect of noise. For each given $\mathcal{G}$ we will compute the mean and the variance of the distribution of $b_U$ over $\mathcal E$, namely $\langle b_U\rangle_{\mathcal G}$ and $\Var_{\cal G} (b_U)$. It turns out that, see \cite{10.21468/SciPostPhys.10.3.076}
\ba\label{isoformulae}
\langle b_U\rangle_{\mathcal G} &=& \tr{T_2(B_0\ot\omega_0)\mathcal{R}_{\cal G}^{(2)}(U^{\ot1,1})},\\
\Var_{\cal G} (b_U) &=& \langle b_U^2\rangle_{\cal G} - \langle b_U\rangle_{\mathcal G}^2,
\ea
where $b_U^2=\tr{T_{(13)(24)}\mathcal{R}_{\cal G}^{(4)}(U^{\ot2,2})(\omega_0^{\ot2}\ot B_0^{\ot 2})}$ and the isospectral twirling of order $k$ of a unitary operator $U$ is defined as:
\ba
\mathcal{R}_{\cal G}^{(2k)}(U):=\int_{\cal G}\,d\mu_GG^{\dag\ot 2k}U^{\ot k,k}G^{\ot 2k},
\ea
where $U^{\ot k,k}=U^{\ot k}\ot U^{\dag\ot k}$ and the $T_\pi$ are operators for the permutations of $2k$ objects $\pi$. In practice, the isospectral twirling of order $k$ of an operator $U$ is the order $2k$ moment of the operator $U$. The result of this operation is an operator whose spectrum is the same as that of $U$, where the eigenvectors have been averaged over all elements of the group $\mathcal{G}$. 
The details of the evaluation of $\mathcal{R}_{\cal G}^{(2k)}(U)$ for several instances of $\mathcal G$ and $U$ are given in~\ref{isodetails}.

We now provide some examples. First we consider as core operators a couple of Clifford unitaries, a simple CNOT $U_\mathrm{c}=C_X$ which leaves the state $\ket{00}$ invariant and $U_\mathrm{c}=C_X (H\ot I)$ that prepares the maximally entangled state $| \Phi^+\rangle$. We also consider $U_\mathrm{c}=W(\theta)$ that prepares the maximal violating state for $\theta=\pi/4$. 
For the group $\mathcal{G}$ we consider both the full unitary group $\mathcal{U}$ and the Clifford group $\Cl$ applied symmetrically on both qubits or only on qubit $A$ or qubit $B$. In view of Theorem~\ref{theorem2}, we expect that asymmetric twirling will give better results for the probability of violation. The analytic expressions for the corresponding mean and standard deviation are shown in Table~\ref{tab:isospectral_twirling}. Note that the expressions for the mean coincide for the two groups, but differ for the standard deviation due to the fact that the Clifford group is a 3-design but not a 4-design, and so averages over $\Cl$ and $\mathcal{U}$ coincide up to the third moments~\cite{zhu_clifford_2016}. 
In Figure \ref{fig:isospectral_plot} we plot the means and variances for $U_\mathrm{c}=W(\theta)$ as a function of $\theta$. 

\begin{table}[!h]
    \centering
    \renewcommand{\arraystretch}{1.6}
    \begin{tabular}{|c|c|c|>{\centering\arraybackslash}m{6cm}|}
    \hline
    $\langle\rangle_{\cal G}$ & $C_X$&$C_X(H\otimes I)$&$W(\theta)$\\
    \hline
    $\langle b_U\rangle_{\mathcal{U}} $&$0.2$&$1/15$&$\frac{1}{15} (2 \sin (\theta )+1)$\\
    $\Var_{\cal U}(b_U)$& $0.79$ &$0.80$&$\frac{-22 \sin (\theta )+26 \cos (2 \theta )+5016}{6300}$\\
     $\Var_{\cal C}(b_U)$&$ 0.98$&$1.59$&$\frac{-64 \sin (\theta )+1440 \cos (\theta )+357 \cos (2 \theta )+3937}{3600}$\\
    $\langle b_U \rangle_{\mathcal{U}_A}$&$-2/3$&1/3&$\frac{1}{3} (\sin (\theta )+1)$\\
    $\Var_{\mathcal{U}_A}(b_U)$&$37/45$&$31/45$&$\frac{-7 \sin (\theta )+\cos (2 \theta )+45}{45}$\\
     $\Var_{\mathcal{C}_A}(b_U)$&$8/9$&$19/18$&$\frac{-2 \sin (\theta )-\cos (2 \theta )+3}{36}$\\
     $\langle b_U \rangle_{\mathcal{U}_B}$&$1$&$2/3$&$\frac{2}{3} (\sin (\theta )+\cos (\theta ))$\\
    $\Var_{\mathcal{U}_B}(b_U)$&$0$&$37/45$&$\frac{\sin (2 \theta )+37}{45}$\\
     $\Var_{\mathcal{C}_B}(b_U)$&$0$&$8/9$&$\frac{4\sin (\theta ) \cos (\theta )+8}{9}$\\
    \hline
    \end{tabular}
    
    \caption{Summary of the isospectral twirling for the operators $C_X=\dyad{0}\otimes{I}_B+\dyad{1}\otimes X_B$ and $W(\theta)\equiv (R_y(\theta) \otimes I) C_X (H \otimes I)$. The leftmost column indicates the quantity that is being averaged, together with the group over which the isospectral twirling is performed. The other columns show the corresponding value of the isospectral twirling when the argument of the twirling is the unitary operator indicated at the top of the column. Notice also that the functions correponding to the rightmost column are plotted in Fig.~\ref{fig:isospectral_plot}.}
    \label{tab:isospectral_twirling}
\end{table}
Combining the results in Fig.~\ref{fig:isospectral_plot} with the ones reported in Table~\ref{tab:isospectral_twirling} we observe the following.
\begin{fact}
The mean value of the expectation value over the CHSH operator $\langle b_U\rangle_{\cal G}$ is larger when obtained by twirling asymmetrically on only one qubit as opposed to both qubits symmetrically. The standard deviations tend to be larger, apart from a small range of $\theta$ in Fig.~\ref{fig:isospectral_B}, when using the Clifford group instead of the full Unitary group.
\end{fact}
This confirms once more that asymmetry is indeed a resource in order to violate CHSH, even when one considers the average over unitary groups.

Based on Fig.~\ref{fig:isospectral_plot} one can also make an educated guess of what are the best core unitaries $U_\mathrm{c}$ and groups $\mathcal{G}$  to obtain an ensemble of operators leading to a higher probability of CHSH violations. 
One simply looks for situations where the mean $|\langle b_U\rangle_{\mathcal G} |$ is large and the fluctuations are also large, so that CHSH violations are more likely. 

\begin{figure}[!t]
\begin{subfigure}[b]{0.33\linewidth}
        \centering
        \includegraphics[width=\linewidth]{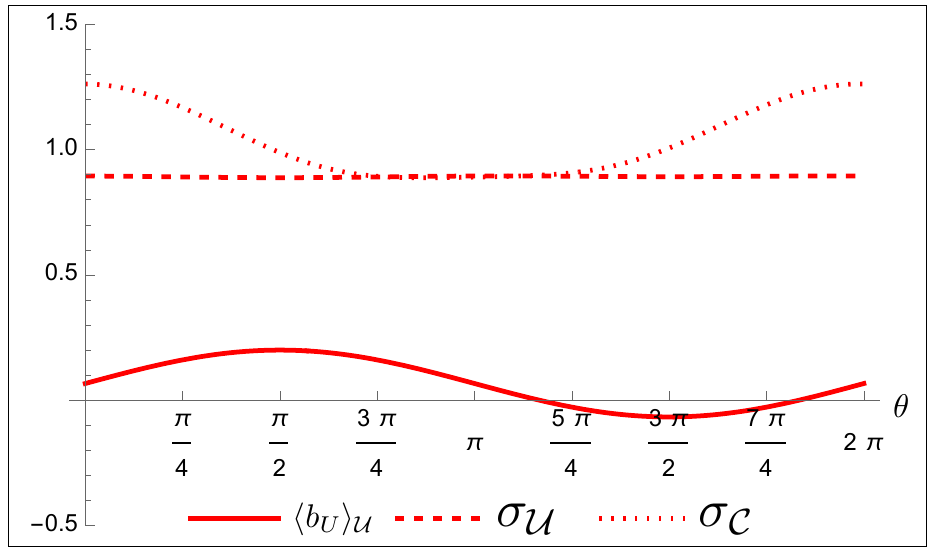}
        \caption{}
        \label{fig:isospectral_whole}
    \end{subfigure}
    \begin{subfigure}[b]{0.33\linewidth}
        \centering
        \includegraphics[width=\linewidth]{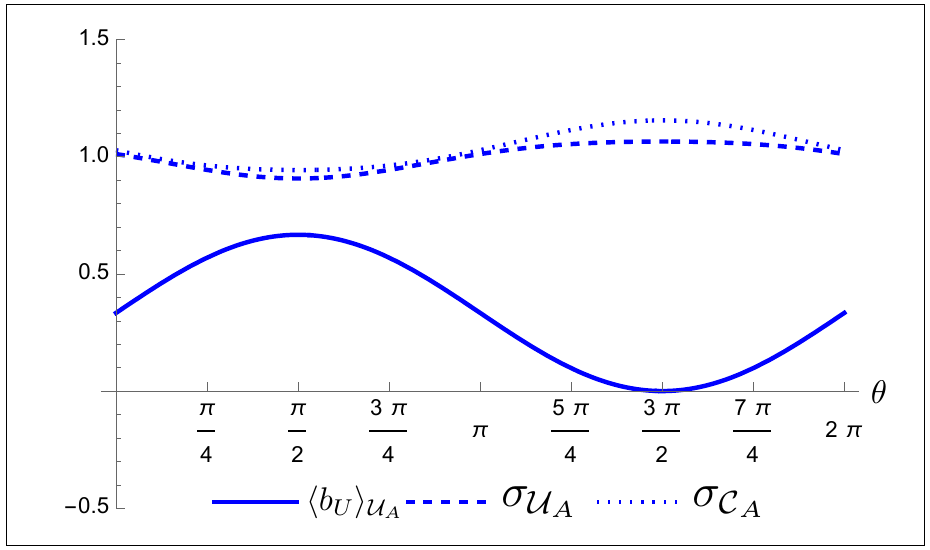}
        \caption{}
        \label{fig:isospectral_A}
    \end{subfigure}
    \begin{subfigure}[b]{0.33\linewidth}
        \centering
        \includegraphics[width=\linewidth]{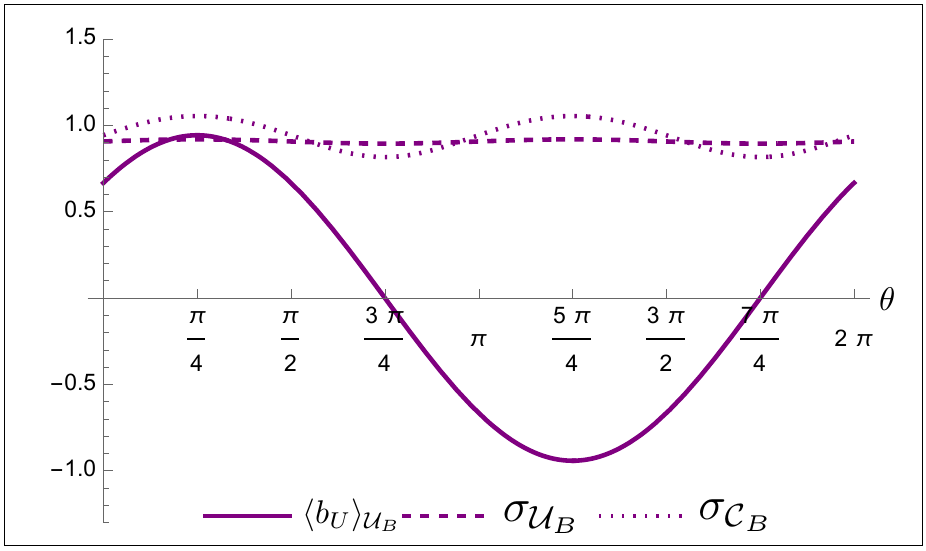}
        \caption{}
        \label{fig:isospectral_B}
    \end{subfigure}
\caption{Plots of the average $\langle b_U\rangle_{\mathcal G}$ and of standard deviations $\sigma_{\mathcal{G}}:=\sqrt{\Var_{\cal G} (b_U)}$
under isospectral twirling over the groups $\mathcal{G}=\mathcal{U},\mathcal{C}$ of the operator $U_{\rm viol}$. In every plot the solid line indicates the average, while the dashed and dotted lines indicate the average standard deviation with respect to the unitary and Clifford group respectively. Panel (a) shows the results when the average is performed over both qubit, while Panel (b) and (c) show the results when the average is performed on only one of the two qubit.}
\label{fig:isospectral_plot}
\end{figure}

\begin{table}
\begin{tabular}{|c|c|c|c|c|c|c|c|}
\hline
$U_c$&$P_{\rm viol}(\mathcal{U})$&$P_{\rm viol}(\mathcal{U}_A)$&$P_{\rm viol}(\mathcal{U}_B)$&$P_{\rm viol}(\mathcal{C})$&$P_{\rm viol}(\mathcal{C}_A)$&$P_{\rm viol}(\mathcal{C}_B)$&$\mathcal{M}_2(U)$\\
\hline
$C_X$&2.2\%&10.8\%&0&0&0&0&0\\
$C_X(H\otimes{I})$&2.5\%&8.3\%&10.8\%&0&0&0&0\\
$W(\pi/2)$&2.2\%&10.9\%&10.8\%&0&0&0&0\\
$W(\pi/3)$&2.3\%&9.6\%&17.2\%&0.3\%&4.2\%&16.6\% & 0.240\\
$
W (\pi/4)$&2.3\%&8.7\%&17.8\%&0.3\%&4.1\%&16.7\% & 0.332\\
$W(\pi/8)$&2.4\%&8.2\%&16.1\%&0.3\%&4.0\%&16.6\% &0.154\\
$\tilde{C}_X$&2.2\%&0&10.8\%&0&0&0&0\\
$\tilde{C}_X({I}\otimes H)$&2.5\%&10.8\%&8.3\%&0&0&0&0\\
$\tilde{W}(-\pi/2)$&2.2\%&10.8\%&10.9\%&0&0&0&0\\
$\tilde{W}(-\pi/3)$&2.3\%&17.2\%&9.6\%&0.3\%&16.6\%&4.2\% & 0.240\\
$\tilde{W}(-\pi/4)$&2.3\%&17.8\%&8.7\%&0.3\%&16.7\%&4.1\% & 0.332\\
$\tilde{W}(-\pi/8)$&2.4\%&16.1\%&8.2\%&0.3\%&16.6\%&4.0\% & 0.154\\
\hline
\end{tabular}
\caption{Probability of violating the CHSH inequality via isospectral twirling. $P_{\rm viol}(\mathcal{G})$ is the probability of violation under isospectral twirling over the group $\mathcal{G}$. Recall the definition for the core $W(\theta)\equiv (R_y(\theta)\otimes{I})C_X(H\otimes{I})$ . The tilde operators are just the corresponding operators with the role of $A$ and $B$ exchanged.}
\label{tab:isospectral_probabilities}
\end{table}

In Table~\ref{tab:isospectral_probabilities} we give the probabilities of violating the CHSH inequality in this setting. Together with $W(\theta=\pi/4)$ we also include results for $\theta = \pi/2,\pi/3,\pi/8$ for comparison. 

Based on the results of Table~\ref{tab:isospectral_probabilities} we can draw several conclusions.

\begin{fact}
Twirling over both qubits symmetrically gives $P\mathrm{viol}$ almost equal to the full Haar result when using $\mathcal{G}=\mathcal{U}$, while it results in very small $P_\mathrm{viol}$ when using $\mathcal{G}=\mathcal{C}$ (of course probabilities are exactly zero when the entire ensemble is made of Clifford unitaries because of Theorem~\ref{theorem1})
\end{fact}
Another interesting result is summarized in the following fact.
\begin{fact}
\label{fact:twirling_control}
Despite the operator $C_X$ preparing the $\ket{00}$ state, the effect of twirling over the control qubit gives fairly large $P_\mathrm{viol}$, while twirling over the idler gives $P\mathrm{viol}=0$.
\end{fact}
Fact~\ref{fact:twirling_control} can be understood by looking at Table~\ref{tab:isospectral_twirling} since the values of the mean and variance are large when twirling over the control qubit.

One can then observe that the highest value for $P_{\mathrm{viol}}$ is obtained using $W(\pi/4)$ as core operator. This is to be expected, as this corresponds to a maximal injection of non-stabilizerness. However, fairly large error in the angle $\theta$ from the optimal value $\pi/4$ produce similarly good results. This leads us to the next fact.
\begin{fact}
The probability of violation does not strongly depend on the amount of non-stabilizerness, but rather on the asymmetry in the twirling.
\end{fact}

Then, by looking at the probability of violations when the twirling is performed over the Clifford group, one can observe a similar behavior as the one observed for the full unitary group, namely that asymmetric twirling leads to larger probabilities. However, we can further note another fact.
\begin{fact}
\label{fact:clifford_twirling}
Performing the isospectral twirling using the Clifford group on only one qubit  gives a probability of violation that is comparable with the one obtained by averaging over the full unitary group. 
\end{fact}
Fact~\ref{fact:clifford_twirling} is kind of surprising, as it implies that random Clifford operations are as efficient as random unitaries when it comes to violate the CHSH inequality, provided the core operator has the necessary non-stabilizer resources.

Summarizing, by means of the isospectral twirling we have shown how to build  ensembles of unitary operators with large probability of non-local violations. Our findings are of course consistent with the theorems we proved, namely that in absence of non-stabilizer resources it is impossible to violate the CHSH inequality as well as the fact that resources must be {\it asymmetric} with respect to the qubits. 

\section{Conclusions and outlook\label{sec:conclusions}}

The interplay between non-stabilizer   and entanglement resources is fundamental for the violation of the CHSH inequality. In this paper, we showed the resources involved in such violations as they are encoded in the quantum operations performed on the initial preparation before a resource-free reference measurement.

We showed that, in order to obtain a violation it must be both entangling and have non stabilizing power. Moreover, it must be asymmetric and - surprisingly - the non-stabilizer resource SE must be local. We compute the probability of violation given the resources. Then, employing results from representation theory, we systematically prepare ensembles of unitary evolutions that provide higher probability of violation. These techniques represent a modelization in quantum control, where one has limited control over the evolution one can implement in the system. 

In perspective, one would be interested in knowing how the proposed setting and techniques can be employed to study higher dimensional systems and study other fundamental probes of quantumness such as quantum discord and contextuality. Finally, we plan to study Hardy non-locality~\cite{PhysRevLett.68.2981}, and possibly other settings of non-locality~\cite{PhysRevLett.108.100402,PhysRevLett.132.250205}, from the resource-theoretic point of view by extending the proposed setting to general quantum operations. 

\section*{Acknowledgements}
This research was funded by the Research Fund for the Italian Electrical System under the Contract Agreement "Accordo di Programma 2022–2024" between ENEA and Ministry of the Environment and Energetic Safety (MASE)- Project 2.1 "Cybersecurity of energy systems". AH, JO and Stefano Cusumano acknowledge support from the PNRR MUR project PE0000023-NQSTI. AH acknowledges support from the PNRR MUR project CN 00000013-ICSC. JO acknowledges ISCRA for awarding this project access to the LEONARDO super-computer, owned by the EuroHPC Joint Undertaking, hosted by CINECA (Italy) under the project ID: PQC- HP10CQQ3SR. AH thanks M. Howard for interesting discussions and comments. The authors wish to thank the anonymous Referee for their comments and suggestions.

\appendix

\section{Proof of Theorem \ref{TheoremI}}\label{appendixxx}
We provide the algebraic derivation of the functional forms for the maximal CHSH violation $C(t)$ and the non-local non-stabilizerness $M_{NL}(t)$, with $t := \sin^2(2\theta)$, and we prove that the optimal coefficient in Eq.~\eqref{Eq. C(t)Inequality} is $\alpha_{\text{opt}} = \ln(2)/\sqrt{2}$.

We first justify the restriction of the Schmidt parameter $\theta$ to the interval $[0, \pi/4]$. Since the Schmidt decomposition is invariant under permutations of the Schmidt coefficients, the states $\cos\left(\theta\right) \ket{00}+\sin\left(\theta\right)\ket{11}$ and $\sin\left(\theta\right)\ket{00}+\cos\left(\theta\right)\ket{11}$ are locally unitarily equivalent. Consequently, any local-unitary invariant depending only on the Schmidt spectrum satisfies: 
\begin{equation}
    F(\theta) = F\left( \frac{\pi}{2} - \theta \right).
\end{equation}
In particular, this holds for both $C(\theta)$ and $M_{NL}(\theta)$.
Therefore, every local-unitary equivalence class admits a representative with $\theta \in [0, \pi/4]$, and we may restrict the analysis to this interval without loss of generality. 

Since Eq.~\ref{Eq. C(t)Inequality} becomes an equality at $t=1$, the largest admissible value of $\alpha$ is fixed by the limiting behavior of $g(t)$ as $t \rightarrow 1$. Setting $t=1-\epsilon$ for $\epsilon \rightarrow 0^+$, we apply a Taylor expansion:
\begin{align}
    \log(1-\epsilon+\epsilon^2) &= -\frac{\epsilon}{\ln(2)} + O(\epsilon^2), \\
2\sqrt{1+t} - 2\sqrt{2} &= 2\sqrt{2-\epsilon} - 2\sqrt{2} = -\frac{\epsilon}{\sqrt{2}} + O(\epsilon^2).
\end{align}
Evaluating the limit of the ratio:
\begin{align}
\lim_{t \to 1} g(t) = \frac{-\epsilon/\sqrt{2}}{-\epsilon/\ln(2)} = \frac{\ln(2)}{\sqrt{2}} = \alpha_{\rm opt}.
\end{align}
To prove that $\alpha_{\text{opt}}$ is sufficient for the bound to hold globally, we define the gap function: 
\begin{align}
    h_{\alpha}(t) = 2\sqrt{2} + \alpha \log(1-t+t^2) - 2\sqrt{1+t}.
\end{align}
Since, by construction, $h_{\alpha_{\text{opt}}}(1)=0$, it suffices to show that $h_{\alpha_{\text{opt}}} (t)$ is non-increasing on $[0,1]$, which implies $h_{\alpha_{\text{opt}}}(t)\geq 0$ for all $t\in[0,1]$. The derivative is: 
\begin{align}
 h_{\alpha}'(t) = \frac{\alpha}{\ln(2)} \frac{2t-1}{1-t+t^2} - \frac{1}{\sqrt{1+t}}.  
\end{align}
Substituting the optimal value $\alpha_{\text{opt}}$: 
\begin{align}
    h^{'}_{\alpha_{\text{opt}}} (t)= \frac{2t - 1}{\sqrt{2} (1-t+t^2)} - \frac{1}{\sqrt{1+t}}.
\end{align}
For $t \in [0, 1/2[$, it holds $2t-1 < 0$, ensuring $h^{'}_\alpha (t) < 0$. For $t \in [1/2, 1]$, by cross-multiplying, the condition $h^{'}_\alpha(t) \leq 0$ is equivalent to  
\begin{align}
    (2t-1)\sqrt{1+t} \leq \sqrt{2}(1-t+t^2).
\end{align}
Notice that the terms $(1-t+t^2)$ and $\sqrt{1+t}$ are positive. Squaring both sides gives: 
\ba
\nonumber
    (2t-1)^2(1+t) \leq 2(1-t+t^2)^2 &&\iff (4t^2-4t+1)(1+t) \leq 2(1-2t+3t^2-2t^3+t^4) \\
    \nonumber
    &&\iff 4t^3-3t+1 \leq 2-4t+6t^2-4t^3+2t^4 \\
    &&\iff 0 \leq 2t^4-8t^3+6t^2-t+1 =: A(t).
\ea
To prove $A(x) \geq 0$, we examine its derivatives: 
\begin{align}
    A'(t) = 8t^3 - 24t^2 + 12t - 1, \quad A''(t) = 12(2t^2 - 4t + 1).
\end{align}
The roots of $A^{''}(t)$ are $1 \pm \sqrt{2}/2$. On $[1/2, 1]$, $A^{''} < 0$, so $A'(t)$ is strictly decreasing. Since at the left endpoint it holds $A'(1/2)=0$, it follows $A^{'}(t) \leq 0$ on the interval. Therefore, $A(t)$ decreases toward $A(1) = 0$, confirming $A(t) \geq 0$. This implies that $h^{'}_{\alpha_{\rm opt}}(t) \leq 0$  throughout the interval $[0,1]$, and therefore $h_{\alpha_{\rm opt}}(t)$ is monotonically non-increasing. Evaluating the function at the right endpoint yields: 
\begin{align}
    h_{\alpha_{\rm opt}}(1) = 2\sqrt{2} + \frac{\ln(2)}{\sqrt{2}} \log_2(1) - 2\sqrt{2} = 0.
\end{align}
Since the function decreases monotonically toward a minimum value of $0$ at $t=1$, it follows that $h_{\alpha_{\rm opt}}(t) \geq 0$ for all $t \in [0,1]$. This confirms that the bound of Eq.~\eqref{Eq. C(t)Inequality} is  satisfied for all $t \in [0,1]$, completing the proof of sufficiency for the optimal coefficient. 

\section{Isospectral twirling\label{sec:appendix}}

\subsection{Fluctuations of \texorpdfstring{$b_{U}$}{bU} for the Haar measure}
In order to obtain an estimate of the probability of violation through the Chebyshev inequality, we need to compute also the fluctuations, as given by the variance $\Var_U(b_U)$, of the expectation value of $B_0$. This is defined as:
\ba
\label{eq:standard_deviation}
\Var_U(b_U)=\left\langle\tr{B_0U\omega_0U^\dag}^2\right\rangle_U-\left\langle\tr{B_0U\omega_0U^\dag}\right\rangle_U^2=\left\langle\tr{B_0U\omega_0U^\dag}^2\right\rangle_U,
\ea
so that we need to compute:
\ba \langle b_U^2\rangle_{U}=\left\langle\tr{B_0U\omega_0U^\dag}^2\right\rangle_U=\int_{\cal U}\,d\mu_U\tr{B_0U\omega_0U^\dag}^2.
\ea
We can rewrite the argument of the Haar average as:
\ba
b_U^2=\tr{B_0U\omega_0U^\dag}^2=\tr{B_0U\omega_0U^\dag\ot B_0U\omega_0U^\dag}=\tr{B_0^{\ot2}(U^{\ot2}\omega_0^{\ot2}U^{\dag\ot2})}.
\ea
Inserting this back in the average one gets:
\ba
\nonumber
\int_{\cal U}\,d\mu_U\tr{B_0U\omega_0U^\dag}^2=\int_{\cal U}\,d\mu_U\tr{B_0^{\ot2}(U^{\ot2}\omega_0^{\ot2}U^{\dag\ot2})}=\tr{B_0^{\ot2}\int_{\cal U}\,d\mu_U(U^{\ot2}\omega_0^{\ot2}U^{\dag\ot2})}.\\
\ea
We can see that $\Var_U(b_U)$ depends on the second moment of the state $\omega_0$. The second moment of an operator $O$ is given by: 

\ba
\label{eq:second_moment}
\mathcal{R}_{\cal G}^{(2)}(O)=\frac{1}{d^2-1}\left[\left(\tr{O}-d^{-1}\tr{T_2O}\right)I+\left(\tr{T_2O}-d^{-1}\tr{O}\right)T_2\right].
\ea
In case $O$ is the two-fold copy of a state, i.e. $O=\omega_0^{\ot2}$, one obtains:
$\int_{\cal U}\,d\mu_U(U^{\ot2}\omega_0^{\ot2}U^{\dag\ot2})=(I+T_2)/20$. 
Putting this back into the trace we get
$\left(\tr{B_0^{\ot2}I}+\tr{B_0^{\ot2}T_2}\right)/ 20=\tr{{I}_4}/5=4/5$
so that $\sigma_{B_0}=2/{\sqrt{5}}$.

Having the variance one can apply the Chebyshev inequality to obtain a first, rough, upper bound to the probability of violating the Bell inequality. The Chebyshev inequality states that the probability of a random variable $X$ to differ from its average value $\bar{X}$ is bounded as $P(\left|X-\bar{X}\right|\geq k\sigma)\leq{1}/{k^2}$. 
In our case the average value is $0$, while $\sigma=2/\sqrt{5}$, which means that in order for $\tr{B_0U\omega_0U^\dag}$ to be greater than $2$ one has to set $k=\sqrt{5}$, so that:
\ba
P(|\tr{B_0U\omega_0U^\dag}|\geq2)\leq\frac{1}{5}. \label{eq:chebishev_violation}
\ea

\subsection{Averages from representation theory}
\label{isodetails}
Let us now go back to the expectation value of $B_0$ and rewrite it as:
\ba
\tr{B_0U\omega_0U^\dag}=\tr{T_2(B_0U\ot\omega_0U^\dag)}=\tr{T_2(B_0\ot\omega_0)U^{\ot1,1}},
\ea
where we applied the swap trick $\tr{AB}=\tr{T_2(A\ot B)}$. Applying the isospectral twirling to the expression above we get:
 \ba
\langle b_U\rangle_{\mathcal{G}}=\langle\tr{B_0U\omega_0U^\dag}\rangle_{\mathcal{G}}=\tr{T_2(B_0\ot\omega_0)\mathcal{R}_{\cal G}^{(2)}(U^{\ot1,1})}.
\ea
We thus see that the average value of $B_0$ depends on the $k=1$ isospectral twirling, which corresponds to the second moment operator of $U^{\ot1,1}$. We can then apply Eq.~\eqref{eq:second_moment} and compute the $k=2$ moment operator of $U^{\ot1,1}$, which is given by:
\ba
\nonumber
\mathcal{R}^{(2)}_{\cal U}(U^{\ot1,1})=\frac{1}{d^2-1}\left[\left(\tr{U^{\ot1,1}}-d^{-1}\tr{T_2U^{\ot1,1}}\right)I+\left(\tr{T_2U^{\ot1,1}}-d^{-1}\tr{U^{1,1}}\right)T_2\right].\\
\ea
The two traces can be easily evaluated:
\ba
\tr{U^{1,1}}=\left|\tr{U}\right|^2=c_2(U)=\sum_{i,j}u_iu_j^*,\\
\tr{T_2U^{\ot1,1}}=\tr{UU^\dag}=\tr{I}=d.
\ea
The function $c_2(U)$ is the two point spectral form factor of the unitary operator $U$, and as it can be seen, it only depends on the spectrum of $U$. Plugging these expressions back into the one for the isospectral twirling we obtain:
\ba
\mathcal{R}^{(2)}_{\cal U}(U^{\ot1,1})=\frac{c_2(U)-1}{15}I+\frac{16-c_2(U)}{60}T_2.
\ea
We can plug this expression back into the expectation value, obtaining:
\ba
\langle b_U\rangle_{\mathcal{U}}=\tr{T_{(12)}(B_0\ot\omega_0)\mathcal{R}_{\cal U}^{2}(U^{\ot1,1})}=\frac{c_2(U)-1}{15}\tr{T_2\left(B_0\otimes\omega_0\right)}=\frac{c_2(U)-1}{15}.
\ea
This expression only depends on the 2 point spectral form factor, and can thus be trivially upper bounded considering $|c_2(U)|\leq d^2=16$, leading to 
\ba
\left|\langle b_U\rangle_{\cal U}\right|\leq1.
\ea
While this upper bound is still quite far from the one needed to observe violations of locality, two points must be noted. First, the upper bound is anyway better than the value obtained with the Haar averaged state. Second, this quantity only depends on the spectrum of $U$, and thus we can hope that, in presence of large enough fluctuations, one can optimize the choice of the unitary  $U$ in order to obtain a larger probability of violating the CHSH inequality.

To pursue this path we need to compute the isospectral twirling of the variance, as defined in Eq.~\eqref{eq:standard_deviation}. In practice we need to compute
\ba
\nonumber
\langle b_U^2\rangle=\left\langle\tr{B_0U\omega_0U^\dag}^2\right\rangle_{\cal G}&=&\left\langle\tr{T_{(13)(24)}U^{\ot2,2}(\omega_0^{\ot2}\ot B_0^{\ot 2})}\right\rangle_{\cal U}\\
&=&\tr{T_{(13)(24)}\mathcal{R}_{\cal G}^{(4)}(U^{\ot2,2})(\omega_0^{\ot2}\ot B_0^{\ot 2})}.
\ea
So, in order to compute the standard deviation under isospectral twirling we need to compute the moment of order $k=4$ of the unitary operator $U^{\ot2,2}$. At this point we must note that one is not forced to average over the whole unitary group, but in principle it is also possible to perform the isospectral twirling over the Clifford group. As the latter is known to form a 3-design~\cite{zhu_clifford_2016}, it is clear that we are going to observe a difference only when computing the standard deviation, which involves the fourth order average.

The fourth order average over the whole unitary group can be easily evaluated with the same techniques used for the $k=1,2$ moment operators, leading to the final result:
\ba
\tr{T_{(13)(24)}\mathcal{R}_{\cal U}^{(4)}(U^{\ot2,2})\left(\omega_0^{\ot2}\ot B_0^{\ot2}\right)}=\frac{1344-4c_2(U)+\tilde{c}_2(U)+2\Re[c_3(U)]+c_4(U)}{1680}.
\ea
This expression once again depends only on the spectrum of $U$, but this time the expression features also the three and four points spectral form factor, defined as:
\ba
\tilde{c}_2(U)=\sum_{i,j}u_i^2u_j^{*2},\qquad c_3(U)=\sum_{i,j,k}u_iu_ju_k^{*2},\qquad c_4(U)=\sum_{i,j,k,\ell}u_iu_ju^*_ku^*_\ell.
\ea
The variance under isospectral twirling $\Var_{\mathcal{U}}(b_U)$ is then worth:
\ba
\label{eq:isospectral_std_dev}
\Var_{\mathcal{U}}(b_U)=\frac{4 \left(41-28 c_2\right) c_2+15 \tilde{c}_2+30 \Re\left(c_3\right)+15c_4+20048}{25200}.
\ea

Let us now turn to the Clifford average. The formula to perform this average has been already derived~\cite{PhysRevLett.121.170502}, and it reads:
\ba
\nonumber
\mathcal{R}_{\cal C}^{(4)}(U)=\int_{\cal C}d\mu_C C^{\dag\ot4}U^{\ot2,2}C^{\ot4}=\sum_{\pi,\sigma\in S_4}W^+_{\pi\sigma}\tr{T_\pi QU^{\ot2,2}}QT_\sigma+W^-_{\pi\sigma}\tr{T_\pi Q^{\perp}U^{\ot2,2}}Q^{\perp}T_\sigma,\\
\label{eq:4_clifford}
\ea
where $Q^\perp=I^{\otimes 4}-Q$, and the Weingarten coefficients $W^\pm$ can be computed in terms of the characters of the symmetric group representations as:
\ba
W^\pm_{\pi\sigma}=\sum_{\lambda}\frac{d_\lambda^2}{(4!)^2}\frac{\chi^\lambda(\pi\sigma)}{D^{\pm}_\lambda},
\ea
where $\lambda$ labels the irreducible representations of the symmetric group $S_4$, $d_\lambda$ is the dimension of the corresponding irrep, $\chi_\lambda(\pi\sigma)$ is the character of the corresponding permutation and $D^+_\lambda=\tr{P_\lambda Q}$, $D^-_\lambda=\tr{P_\lambda Q^\perp}$, $P_\lambda$ being the projector onto the irrep $\lambda$.

One can then use Eq.~\eqref{eq:4_clifford} to compute the variance under isospectral twirling over the Clifford group to obtain:
\ba
\Var_{\mathcal{C}}(b_U)=\frac{7}{9}+\frac{|c_{U^2}|^2+2\Re[c_U^{*2}c_{U^2}]+|c_U|^4}{80}-\frac{|c_{UU^\dag}|^2+c_{(UU^\dag)^2}}{72}-\left(\frac{c_2(U)-1}{15}\right)^2,
\ea
where we have defined:
\ba
|c_U|^4=&&d^{-2}\sum_P|\tr{PU}|^4=d^{-2}\sum_{P}\sum_i|e^{-i\phi_i}\bra{\phi_i}U\ket{\phi_i}|^4,\\
|c_{U^2}|^2&=&d^{-2}\sum_P|\tr{PUPU}|^2=d^{-2}\sum_{P}\sum_{i,j}|e^{-i(\phi_i+\phi_j)}\bra{\phi_j}U\dyad{\phi_i}U\ket{\phi_j}|^2,\\
\nonumber
c_U^{*2}c_{U^2}&=&d^{-2}\sum_P\tr{PU^\dag}^2\tr{PUPU}\\
&=&d^{-2}\sum_{P}\sum_{i,j,k}e^{-i(\phi_i+\phi_j-2\phi_k)}\bra{\phi_j}U\dyad{\phi_i}U\ket{\phi_j}\bra{\phi_k}U^\dag\ket{\phi_k}^2,\\
|c_{UU^\dag}|^2&=&d^{-2}\sum_P|\tr{PUPU^\dag}|^2=d^{-2}\sum_P\sum_{i,j}|e^{-i(\phi_i-\phi_j)}\bra{\phi_j}U\dyad{\phi_i}U^\dag\ket{\phi_j}|^2,\\
\nonumber
c_{(UU^\dag)^2}&=&d^{-2}\sum_P\tr{UPU^\dag PUPU^\dag}\\
&=&d^{-2}\sum_P\sum_{i,j,k,\ell}e^{-i(\phi_i+\phi_j-\phi_k-\phi_\ell)}\bra{\phi_i}U\dyad{\phi_j}U^\dag\dyad{\phi_k}U\dyad{\phi_\ell}U^\dag\ket{\phi_i}.
\ea
where we have written the unitary operator $U$ as $\sum_ie^{-i\phi_i}\dyad{\phi_i}$ where $\ket{\phi_i}$ are eigenvectors of $U$ and $e^{-i\phi_i}$ the corresponding eigenvalues.
Notice that in contrast with the average over the unitary group, in this case it is not possible to give a closed form of the standard deviation in terms of the spectrum of $U$ alone. Indeed, the isospectral twirling over the Clifford group does depend on the matrix elements of $U$ in the Pauli basis. This means that in order to evaluate the standard deviation in this case we need the full expression of the unitary operator $U$.

\bibliographystyle{unsrtnat}
\bibliography{bell_CHSH_biblio}
\end{document}